\documentclass[pra,twocolumn,showpacs,preprintnumbers,nofootinbib,amsmath,amssymb]{revtex4}

\usepackage{graphicx, color, graphpap}
\usepackage{amsmath}
\usepackage{amssymb}
\usepackage{amsthm}
\usepackage{pstricks}
\usepackage{float}
\usepackage{multirow}
\usepackage{hyperref}
\long\def\ca#1\cb{} 

 \def\outl#1{} \def\np{} \def\xa{} \def\xb{}  

\ca
 \def\outl#1{\par{\medskip\noindent\hspace*{.5cm}\bf
      \mathversion{bold}#1\mathversion{normal}\smallskip} }
 \long\def\xa#1\xb{}
 \def\np{\newpage }
\cb

\ca
 \def\outl#1{\par{\medskip\noindent\hspace*{.5cm}\bf
      \mathversion{bold}#1\mathversion{normal}\smallskip} }
\def\np{} \def\xa{} \def\xb{}  
\cb

\newcommand{\avg}[1]{\langle #1\rangle }
\newcommand{\ket}[1]{|#1\rangle}               
\newcommand{\colo}{\,\hbox{:}\,}              
\newcommand{\bra}[1]{\langle #1|}              
\newcommand{\dya}[1]{\ket{#1}\bra{#1}}
\newcommand{\dyad}[2]{\ket{#1}\bra{#2}}        
\newcommand{\ip}[2]{\langle #1|#2\rangle}      

\newcommand{\EC}{\mathcal{E}}
\newcommand{\FC}{\mathcal{F}}
\newcommand{\GC}{\mathcal{G}}
\newcommand{\HC}{\mathcal{H}}
\newcommand{\IC}{\mathcal{I}}

\newcommand{\LC}{\mathcal{L}}

\newcommand{\Tr}{{\rm Tr}}
\renewcommand{\geq}{\geqslant}
\renewcommand{\leq}{\leqslant}

\newcommand{\ot}{\otimes}
\newcommand{\ad}{^\dagger}
\newcommand{\trp}{^\textup{T}}

\newcommand{\al}{\alpha }

\newcommand{\Dl}{\Delta}

\newcommand{\lm}{\lambda }

\newcommand{\sg}{\sigma }

\newcommand{\Om}{\Omega }

\newtheoremstyle{example}{\topsep}{\topsep}%
{}
{}
{\bfseries}
{.}
{   }
{\thmname{#1}\thmnumber{ #2}}
\theoremstyle{example}
\newtheorem{example}{Example}

\newtheorem{theorem}{Theorem}
\newtheorem{lemma}[theorem]{Lemma}
\newtheorem{corollary}[theorem]{Corollary}

\begin{document}


\title{Information-theoretic treatment of tripartite systems and quantum channels}
\author{Patrick J. Coles}
\email{pcoles@andrew.cmu.edu}
\author{Li Yu}
\email{liy@andrew.cmu.edu}
\author{Vlad Gheorghiu}
\email{vgheorgh@andrew.cmu.edu}
\author{Robert B. Griffiths}
\email{rgrif@andrew.cmu.edu}
\affiliation{Department of Physics, Carnegie Mellon University, Pittsburgh, Pennsylvania 15213, USA}

\begin{abstract}
A Holevo measure is used to discuss how much information about a given POVM on system $a$ is present in another system $b$, and how this influences the presence or absence of information about a different POVM on $a$ in a third system $c$. The main goal is to extend information theorems for mutually unbiased bases or general bases to arbitrary POVMs, and especially to generalize ``all-or-nothing" theorems about information located in tripartite systems to the case of \emph{partial information}, in the form of quantitative inequalities. Some of the inequalities can be viewed as entropic uncertainty relations that apply in the presence of quantum side information, as in recent work by Berta et al.\ [Nature Physics 6, 659 (2010)]. All of the results also apply to quantum channels: e.g., if $\EC$ accurately transmits certain POVMs, the complementary channel $\FC$ will necessarily be noisy for certain other POVMs. While the inequalities are valid for mixed states of tripartite systems, restricting to pure states leads to the basis-invariance of the difference between the information about $a$ contained in $b$ and $c$.
\end{abstract}

\maketitle

\xa

\xb
\np
\section{Introduction}\label{sct1}
\xa

A significant part of current quantum information research can be understood
as an attempt to find answers to the following question: How much of what kind
of information about what is located where?  In this paper we provide specific
answers to these questions in the case of a general tripartite quantum system:
subsystems $a$, $b$, and $c$ are described by some sort of quantum state
(pre-probability) that induces a joint probability distribution on different
properties of these systems.  Appropriate statistical correlations can then be
thought of in terms of, for example, system $b$ containing information of some
sort about certain physical properties of system $a$.  To discuss how much
information of this kind is contained in or can be found in $b$ requires
some sort of quantitative measure, and it is natural to look for something
resembling the well-known Shannon measures in classical information; see
\cite{CvTh06} for a modern introduction to this subject.

Although it is rather natural to treat systems $a$, $b$ and $c$ on an equal
footing---and that is the perspective of this paper---one can also think of
the properties as existing at different times.  For example, $a$ might be the
entrance to a quantum channel with $b$, possibly but not necessarily the same
physical system, the output of the channel and $c$ the ``environment'' at this
later point in time.  Such a \emph{dynamical} perspective is well-known in
classical information theory as it applies to a noisy channel, where it can be
discussed using the same information measures, e.g., the mutual information
$H(X\colo Y)$, that apply to statistically-correlated systems (think of a
shared key used for cryptographic purposes) at the same time, a \emph{static}
perspective.  Both perspectives are also possible for problems in
quantum information theory, though this has not received as much attention as
we think it deserves, and viewing expressions which are formally the same (or
closely related) from distinct points of view can make a valuable contribution
to one's intuitive understanding of a situation. 

Of course, quantum information theory is more general than classical
information theory, so the conceptual ideas provided by the latter are
insufficient for discussing the quantum world.  In this paper we take the
perspective that a valuable way to think about the quantum case is to
distinguish different \emph{types} or \emph{species} of quantum information
\cite{Gri07}.  For example, if $a$ is a single qubit the distinction between
$\ket{0}$ and $\ket{1}$ constitutes the ``$z$'' type of information, whereas
the distinction between $\ket{+}$ and $\ket{-}$, with $\ket{\pm} =
(\ket{0}\pm\ket{1})/\sqrt{2}$, is the ``$x$ type.''
Each type by itself, even when it refers to
microscopic (thus ``quantum'') properties, follows the usual rules of
classical information theory.  This allows one to immediately transfer a large
body of mathematical formalism and associated physical intuition from the
classical to the quantum domain without risk of falling prey to
inconsistencies and paradoxes.  The quantum nature of the microscopic world
then manifests itself through the fact that incompatible types of information,
corresponding to non-commuting projective decompositions of the identity,
cannot be combined: this is the single framework rule (see, e.g., Ch.~16 of
\cite{CQT}) that allows a fully-consistent use of probabilities in the quantum
domain.\footnote{It is important to note that a type of information as defined here refers
primarily to a \emph{microscopic} quantum property rather than the outcome of
a measurement.  A correctly constructed measurement apparatus can reveal the
property of a microscopic system, so that, for example a Stern-Gerlach
apparatus followed by detectors can determine if the spin-half particle
entering the apparatus had $S_z=+\hbar/2$ or $-\hbar/2$, corresponding to the
qubit states $\ket{0}$ or $\ket{1}$, and in this case the $z$ information
initially possessed by the particle is translated into distinct macroscopic
apparatus states, making the $z$ information ``visible'' or ``classical.''
(For an important application to quantum information theory of the idea that a
macroscopic quantum outcome reveals a prior microscopic state, see \cite{GriNiu96}; for a detailed discussion of the measurement process in
fully quantum terms, see Chs.~17 and 18 of \cite{CQT}.)  However, the concept of $z$ information can also be used in
situations, such as when a qubit is just entering a quantum channel, where
trying to relate it to a measurement, at least as a physical process occurring
at that point in time, is not very helpful.}

In this paper we generalize the notion of a type of quantum information so
that it includes not only a projective decomposition of the identity, a set of
projectors that sum to the identity, but also a general POVM, a collection of
positive operators that sum to the identity. The idea, discussed in Sec.~\ref{sbct2.1}, is that while the operators in a POVM are in general not orthogonal, each corresponds to a projector on a larger Hilbert space, the Naimark extension (which is not unique), and the collection of such projectors sums to the identity on the larger space thus constituting a particular type of quantum information in the sense previously discussed.

The question ``How much?'' has motivated an ongoing search for measures that
extend the very useful idea of entanglement beyond bipartite pure states where
it was first introduced.  Despite a great deal of effort and a large number of
intriguing results \cite{HHHH09}, it seems fair to say that there remain a
large number of unanswered questions even for bipartite mixed states, not to
mention the multipartite case.  It is not obvious that a single number
representing the entanglement, or even a small collection of numbers, will
suffice to embody the physical insights needed for a better understanding of
such systems.  In this paper we introduce measures of information that depend
explicitly on the (quantum) type of information one is considering, so we can
address the question of, for example, how well a noisy quantum channel
performs for different types of input.  The definitions and a detailed
discussion of these measures will be found in Sec.~\ref{sct3}; at this point it
suffices to note that they are of the Holevo form using the quantum von
Neumann entropy, though in some cases they can be generalized using other
types of entropy.  

As well as direct quantitative measures of information certain
\emph{differences} in information measures, e.g., the amount of information of
a given type that is in $b$ minus how much is in $c$, are of interest.  We
refer to these as entropy or information \emph{biases}.  It is not without
interest that the \emph{coherent information}
\cite{NieChu00} when expressed in the language of tripartite
systems is (or least can be thought of as) such a bias; see
Sec.~\ref{sbct3.3}.  In Sec.~\ref{sct4} we show that under appropriate
circumstances an information bias will be independent of the type of
information under consideration.

One of the most striking features of quantum information is that if
information of a particular type corresponding to some orthonormal basis $w$
of system $a$ is perfectly present (perfect correlation, no noise) in system
$b$ for the quantum state under discussion, this prevents or excludes a type
of information $v$ corresponding to a basis mutually unbiased (MU) with respect to
$w$---that is, $v$ and $w$ are mutually-unbiased bases (MUBs)---from being
present in a third system $c$.  In Sec.~\ref{sct5} of this paper we
present quantitative generalizations of this and some other ``all-or-nothing''
theorems to situations in which, for example, almost all information of the
$w$ type of information about $a$ is in $b$ and one wants to bound how much
$v$ information, where $v$ is only approximately MU with respect to $w$, can
be present in $c$.

In particular, Theorem~\ref{thm5} in Sec.~\ref{sbct5.1} presents a bound of
this form.  It extends to POVMs an important inequality proved in
\cite{BertaEtAl}, earlier conjectured in \cite{RenesBoileau}, using a
somewhat simpler proof. This extension was also recently proven in \cite{TomRen2010} using smooth entropies; in contrast our proof approach is based on the relative entropy.
  Various consequences, including the application to a
channel and its complementary channel, are worked out in various corollaries.
As well as thinking of this result as a bound on the amounts of two strongly
incompatible (in the sense of almost MUB) types of information about $a$
present in different locations, Theorem~\ref{thm5} constitutes a generalized
entropic uncertainty relation for system $a$ when the coupling to another
system or systems is taken into account (``quantum side information'' in the
sense discussed in \cite{RenesBoileau, DevWin03}).

Several additional quantitative generalizations of all-or-nothing results are given in Secs.~\ref{sbct5.0}, \ref{sbct5.2}, and \ref{sbct5.3}. The all-or-nothing results can be succinctly stated as follows for orthonormal bases $u$, $v$, and $w$ of $a$, where $u$ and $v$ are MU relative to
$w$ (but not necessarily to each other). If the $w$ type of information is perfectly present in $b$, then (1) $\rho_{ac}$ is block diagonal in the $w$ basis (Lemma~\ref{thm4}), (2) the amount of $u$ information in $b$ is equal to the amount of $v$ information in $b$ (Theorem~\ref{thm8}), (3) if the $v$ information is perfectly present in $b$ then there is a perfect quantum channel from $a$ to $b$ (Theorem~\ref{thm10}), (4) if the $w$ information is completely absent from $c$, then no information about $a$ is in $c$: the two are decoupled (Theorem~\ref{thm11}).

The remainder of this paper is organized as follows.  Section~\ref{sct2} is an
introduction to tripartite systems, including the connection with quantum
channels and their complements, and provides details of what we mean by
different types of quantum information.  Various quantitative measures of
information are introduced, and some of their properties discussed, in
Sec.~\ref{sct3}. Our main results, which, as indicated above, provide
quantitative bounds on the location of various types of information in
different systems, occupy Secs.~\ref{sct4} and \ref{sct5}.  Section~\ref{sct6}
relates our work to various other approaches and publications.  A summary,
which provides an overview of how the different theorems are related to each
other, is in Sec.~\ref{sbct7.1}, followed by an indication of issues worth
further exploration in Sec.~\ref{sbct7.2}.  To make the main presentation
compact and easier to follow, all but the very shortest proofs have been
relegated to appendices.

\xb
\section{Systems with three parts}
\label{sct2}
\xa

\xb
\subsection{POVMs and types of information}
\label{sbct2.1}
\xa

Much work in contemporary quantum information theory is devoted to particular
instances of what may be called the \emph{tripartite system problem} defined
in the following way. Let $\HC_{abc}=\HC_a\ot\HC_b\ot\HC_c$ be a tensor
product of Hilbert spaces of dimensions $d_a$, $d_b$, $d_c$, all assumed to be
finite, and let
\begin{equation}
  I_a = \sum_j P_{aj},\quad I_b = \sum_k Q_{bk},\quad
I_c = \sum_l R_{cl}
\label{eqn1}
\end{equation}
be three POVMs, decompositions of their respective identities into \emph{finite} sets of positive operators, hereafter referred to as $P_a$, etc.\footnote{It is sometimes helpful to imagine the three parts as residing in three different places, say
three different laboratories where Alice, Bob, and Carol can carry out
separate preparations and measurements on them.}  
[Note that we use the symbols $a$, $b$, and $c$ as subscripts (but
occasionally on line) to label subsystems, and indices $j$, $k$, $l$, etc.\ to
label the POVM elements.]  What can be said about the joint probability
distribution
\begin{equation}
  \Pr(P_{aj},Q_{bk},R_{cl}) 
=\Tr(P_{aj} Q_{bk} R_{cl} \rho_{abc}),
\label{eqn2}
\end{equation}
where $\rho_{abc}$ is a density operator acting as a pre-probability
(generator of probabilities in the terminology of Sec.~9.4 of \cite{CQT}),
perhaps but not necessarily a projector $\dya{\Om}$ on the pure state
$\ket{\Om}$?  In particular, what is its
information-theoretic significance?  One is, of course, interested in how
these probabilities, and the corresponding marginal distributions such as 
\begin{align}
  \Pr(P_{aj},Q_{bk}) &= \sum_l  \Pr(P_{aj},Q_{bk},R_{cl})=  \Tr_{ab}(P_{aj}Q_{bk}\rho_{ab}),
\label{eqn3}
\end{align}
with $\rho_{ab}$ the partial trace over $\HC_c$ of $\rho_{abc}$, depend upon
the indices $j$, $k$, and $l$. But of equal, or even greater
interest is their dependence upon the \emph{choice of POVMs}
in \eqref{eqn1}. Here quantum theory, in contrast to classical
physics, allows an enormous number of possibilities.

In what follows we shall want to distinguish various different types of POVM.
A \emph{rank-1} POVM is one in which all the positive operators are of rank
1, which is to say proportional to projectors on one-dimensional spaces; we will employ symbols $L$, $M$, $N$ to denote such POVMs. When
all the POVM elements are projectors (orthogonal projection operators) we have
a \emph{projective decomposition} of the identity.  A rank-1 projective
decomposition is associated with an \emph{orthonormal basis}; e.g., the
orthonormal basis $w=\{\ket{w_j}\}$ of $\HC_a$ gives rise to the decomposition
\begin{equation}
  P_{aj} = \dyad{w_j}{w_j}.
\label{eqn4}
\end{equation}
In what follows we use the lower case letters $u$, $v$, and $w$ to denote
orthonormal bases, and where useful add a subscript, e.g., $w_a$, to indicate
the corresponding system or Hilbert space. A second basis $v=\{\ket{v_j}\}$ is
\emph{mutually unbiased} (MU) relative to $w$---the terms \emph{complementary}
or \emph{conjugate} are also in use---thus $v$ and $w$ are mutually unbiased
bases (MUBs), when $|\ip{v_j}{w_k}|=1/\sqrt{d_a}$ is independent of $j$ and
$k$.

Unlike a general POVM, a projective decomposition can be given a simple
physical interpretation: the projectors, or the subspaces onto which they
project, form a quantum sample space: a collection of mutually exclusive
physical properties one and only one of which is true; see Ch.~5 of
\cite{CQT}.  In previous work \cite{Gri07} such a projective decomposition was
called a \emph{type of information}: e.g., $\Pi_a=\{\Pi_{aj}\}$ is a type of
information about the system $a$.  Two types of information $\Pi_a$ and $\Phi_a$
about the same system are \emph{compatible} provided every projector in one
set commutes with every projector in the other set: $\Pi_{aj} \Phi_{ak} = \Phi_{ak}
\Pi_{aj}$ for every $j$ and $k$; otherwise they are \emph{incompatible}.  Two
distinct rank-1 projective decompositions, or the corresponding orthonormal
bases, are necessarily incompatible if they differ by more than simply
relabeling the projectors, and two MUBs are incompatible to the maximum extent
possible. Probabilistic arguments in quantum mechanics cannot combine results
from incompatible decompositions---the \emph{single framework} rule, see
Ch.~16 of \cite{CQT}---without risk of generating contradictions and
paradoxes.

However, in the present paper we generalize the notion of a type of
information about (say) system $a$ to include any POVM $P_a$ when interpreted
using a \emph{Naimark extension}; see \cite{Prs90b,JozsaEtAl03} or Sec.~9-6 of
\cite{Prs93}. Assume that the Hilbert space $\HC_a$ is a subspace of a larger
Hilbert space $\HC_A$, with $E_a$ the operator on $\HC_A$ that projects onto
$\HC_a$.  If $\HC_A$ has been appropriately chosen there is a projective
decomposition $\{\Pi_{Aj}\}$ of its identity $I_A$ such that
\begin{equation}
  P_{aj} = E_a \Pi_{Aj} E_a.
\label{eqn5}
\end{equation}
In addition, one can always arrange that for each $j$ the rank of $\Pi_{Aj}$ is
the same as the rank of $P_{aj}$, though one may need an additional projector,
call it $\Pi_{A0}$, which is orthogonal to $E_a$, so the corresponding $P_{a0}$
is the zero operator.  (It is possible to set things up so that the rank of
$\Pi_{Aj}$ exceeds that of $P_{aj}$, but in light of \eqref{eqn5} the reverse
is impossible.) An important special case used in proving later results is that any rank-1 POVM $N$ on $a$ is equivalent to some rank-1 projective decomposition (orthonormal basis) on $A$ \cite{JozsaEtAl03}. 
One can if one wishes think of $\HC_A$ as a tensor product
$\HC_a\ot \HC_e$, where $\HC_e$ is the Hilbert space of some reference system,
and $\HC_a$ is itself (isomorphic to) the subspace of kets of the form
$\ket{\psi}\ot\ket{e_0}$, with $\ket{e_0}$ a fixed, normalized ket in $\HC_e$.
In this case the density operator on $A$ is $\rho_A=\rho_{ae}=\rho_a\ot\dya{e_0}$, and $E_a$ in \eqref{eqn5} is simply $I_a\ot\dya{e_0}$. Starting with the projective decomposition $\Pi_A$ for the larger system $A$, one can
think of the corresponding positive operators defined in \eqref{eqn5} as
convenient mathematical tools for computing probabilities in those cases in
which the density operator $\rho_A$ has support in the subspace $\HC_a$ onto
which $E_a$ projects.  From this perspective, and using the corresponding
Naimark extensions for $b$ and $c$, one could reformulate the results given in
later sections of this paper in terms of projective decompositions on the
larger Hilbert spaces.  However, the use of POVMs provides in many cases a
simpler mathematical form, and a theorem that refers to an arbitrary POVM
obviously includes projective decompositions as particular cases.
Nonetheless, when thinking in physical or operational terms about the type of
information represented by a general POVM $P_a$ it is helpful to employ its
Naimark counterpart.

The information about a POVM $P_a = \{P_{aj}\}$ is said to be
\emph{completely} or \emph{perfectly present} in system $b$ provided the
conditional density operators
\begin{equation}
  p_j\rho_{bj}=\Tr_a(P_{aj} \rho_{ab});\quad p_j:=\Pr(P_{aj})=\Tr(P_{aj}\rho_a),
\label{eqn6}
\end{equation}
on $\HC_b$ are mutually orthogonal: $\rho_{bj}\rho_{bj'}=0$ for $j\neq j'$.
Conversely, this type of information is (completely) \emph{absent} from $b$
when the conditional density operators $\rho_{bj}$ are identical. One can
visualize this in terms of measurements as follows. Suppose a POVM $\{P_{aj}\}$ measurement is
carried out on system $a$. Can the value of $j$ be deduced by carrying out an appropriate
sort of measurement on system $b$? If the $\rho_{bj}$ are orthogonal to each
other this is clearly possible using (projective) measurements corresponding
to a suitable decomposition $\{Q_{bk}\}$ of $I_b$. But in the other extreme in
which the $\rho_{bj}$ are identical it is clear that no measurement on $b$
will provide any information about $j$. It is worth noting that one obtains the same values in \eqref{eqn6} by replacing the formulas with $p_j\rho_{bj}=\Tr_A(\Pi_{Aj} \rho_{Ab})$ and $p_j=\Tr(\Pi_{Aj}\rho_A)$, so all of our measures in Sec.~\ref{sbct3.2} quantifying the presence of the $P_a$ information in $b$, depending only on $\{p_j,\rho_{bj}\}$, will be unaffected by replacing $P_a$ with its Naimark extension $\Pi_A$.

\xb
\subsection{Quantum channels}
\label{sbct2.2}
\xa

In some sense the most natural way to state the various results given below in
Secs.~\ref{sct4} and \ref{sct5} is in terms of correlations in which all three
parts $a$, $b$, and $c$ are treated, at least formally, in a symmetrical
fashion.  But some of the more interesting applications are to quantum
channels and complementary channels, in which the channel entrance is not
treated in the same way, either formally or intuitively, as the channel
output.   Hence in order to facilitate application of our
results to the case of channels, we provide a brief explanation, using ideas
in \cite{Gri05, BenZyc06}, of why the ``tripartite'' and the ``channel'' problem are not only closely related to
each other, but in some sense identical problems in the case where one
restricts attention to a pure-state pre-probability $\ket{\Om}\in \HC_{abc}$.

\begin{figure}[h]
\begin{center}
\includegraphics{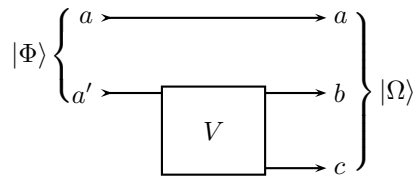}
\caption{%
  How $\ket{\Om}$ is produced by applying the isometry $V$ to an entangled
  state $\ket{\Phi}$.\label{fgr1}}
\end{center}
\end{figure}

Consider the situation shown in Fig.~\ref{fgr1} where
\begin{equation}
\ket{\Om} = (I_a\ot V)\ket{\Phi}.
\label{eqn7}
\end{equation}
is the result of applying an isometry 
\begin{equation}
  V = \sum_j\ket{s_j}\bra{a'_j}
\label{eqn8}
\end{equation}
to the $a'$ part of an entangled state $\ket{\Phi}\in\HC_a\ot\HC_{a'}$, with
$\HC_{a'}$ a copy (i.e., the same dimension) of $\HC_a$.  Here
$\{\ket{a'_j}\}$ is some orthonormal basis of $\HC_{a'}$
held fixed during the following discussion, we are assuming that $d_a\leq d_b
d_c$, and the requirement that $V$ be an isometry, which is to say $V\ad
V=I_a$ is equivalent to the assumption that the kets $\{\ket{s_j}\}$ form an
orthonormal collection spanning the subspace $\HC_s=V\HC_{a'}$ of $\HC_{bc}$.

If, in particular, $\ket{\Phi}$ is the fully-entangled state
\begin{equation}
  \ket{\Phi} =(1/\sqrt{d_a}) \sum_j \ket{a_j}\ot\ket{a'_j},
\label{eqn9}
\end{equation}
then
\begin{equation}
 \ket{\Om} =(1/\sqrt{d_a})  \sum_j \ket{a_j}\ot\ket{s_j},
\label{eqn10}
\end{equation} 
is an example of what in \cite{Gri05} is called a \emph{channel ket},
characterized by the property that
\begin{equation}
  \rho_a =\Tr_{bc}(\dyad{\Om}{\Om}) = I_a/d_a.
\label{eqn11}
\end{equation}
Indeed, given a pre-probability $\ket{\Om}$ such that \eqref{eqn11} holds,
it is necessarily a fully-entangled state on $\HC_a\ot\HC_{bc}$, so it will
have a Schmidt form \eqref{eqn10} for $\{\ket{a_j}\}$ a given orthonormal basis
of $\HC_a$, and using the orthonormal collection of states $\{\ket{s_j}\}$ 
corresponding to this Schmidt decomposition one can define a corresponding
isometry $V$ by means of \eqref{eqn8}.  Thus by employing map-state duality 
(see e.g. \cite{Gri05} or Ch. 11 of \cite{BenZyc06}) one can move from 
a channel ket $\ket{\Om}$ satisfying \eqref{eqn11} to an isometry $V$ or the
reverse. 

From an information-theoretic perspective the isometry $V$ corresponds to
saying that all information about the system $a$ is in the system $bc$, and in
fact in the subspace $\HC_s$ of $\HC_{bc}$ spanned by the $\ket{s_j}$ in
\eqref{eqn10}.  The partial traces onto $\HC_b$ and $\HC_c$ in a sense
``project down'' parts of this information onto these subsystems.  Thus, not
surprisingly, the projector $\Upsilon$ onto $\HC_s$ along with its partial traces
down to $\HC_b$ and $\HC_c$,
\begin{align}
\label{eqn12}
\Upsilon&=V V^\dag=\sum_j\dyad{s_j}{s_j},\nonumber \\
 \Upsilon_b&=\Tr_c(\Upsilon),\quad \Upsilon_c=\Tr_b(\Upsilon),
\end{align}
play useful roles in our thinking about these problems, as they in a sense
describe, in a basis-independent way, how the subspace $\HC_s$ is ``oriented''
relative to the factor spaces $\HC_b$ and $\HC_c$. Note that while $\Upsilon$ is a
projector, $\Upsilon_b$ and $\Upsilon_c$ are positive operators but (in general) not
projectors.

The isometry $V$ in \eqref{eqn8} can be used to define a \emph{quantum
  channel} from $a$ to $b$ through the superoperator
\begin{equation}
  \EC(A) = \Tr_c(VAV\ad) = \sum_l K^{}_l A K_l\ad
\label{eqn13}
\end{equation}
that maps the space $\LC(\HC_a)$ of operators on $\HC_a$ to the corresponding
space $\LC(\HC_b)$ of operators on $\HC_b$. Here the Kraus operators are maps
from $\HC_a$ to $\HC_b$ of the form
\begin{equation}
  K_l = \bra{c_l}V=\sum_j \ip{c_l}{s_j}\bra{a_j},
\label{eqn14}
\end{equation}
where $\{\ket{c_l}\}$ is an orthonormal basis of $\HC_c$, and $\ip{c_l}{s_j}$
is a ket in $\HC_b$, defined in an obvious way, not just a complex number.
Because $V$ is an isometry the Kraus operators satisfy the usual closure
condition
\begin{equation}
  \sum K_l\ad K^{}_l = I_a.
\label{eqn15}
\end{equation}
The \emph{complementary} channel from $a$ to $c$, 
\begin{equation}
  \FC(A) = \Tr_b(V A V\ad) = \sum_m L^{}_m A L_m\ad,
\label{eqn16}
\end{equation}
 is defined in a similar way with Kraus operators
\begin{equation}
  L_m = \bra{b_m}V=\sum_j \ip{b_m}{s_j}\bra{a_j},
\label{eqn17}
\end{equation}
for $\{\ket{b_m}\}$ some orthonormal basis of $\HC_b$, and these again satisfy
the closure condition analogous to \eqref{eqn15}.

The superoperators $\EC$ and $\FC$, and their adjoints $\EC\ad$ and $\FC\ad$
relative to the usual Frobenius inner product $\avg{P,\,Q}=\Tr (P\ad Q)$, can
be expressed directly in terms of $\rho_{abc}=\dyad{\Om}{\Om}$, or its partial
traces such as $\rho_{ab}$, using formulas such as
\begin{align}
  \EC(A) &= d_a\Tr_a[(A\trp\ot I_b)\rho_{ab}],\nonumber \\
  [\EC\ad(B)]\trp &= d_a \Tr_b[(I_a\ot B) \rho_{ab}],
\label{eqn18}
\end{align}
where $\trp$ denotes the transpose relative to the
basis $\{\ket{a_j}\}$ employed in \eqref{eqn9} and \eqref{eqn10}.  The
complete positivity of $\EC$ is equivalent to the requirement that $\rho_{ab}$
be a positive operator; in some respects this is  simpler and more compact
than the traditional definition.  For it to be trace
preserving it is necessary that $\rho_a$ be $I_a/d_a$, \eqref{eqn11}.  Since in
general neither $\rho_b$ nor $\rho_c$ is proportional to the corresponding
identity, the adjoints $\EC\ad$ and $\FC\ad$ are not (in general) trace
preserving, and in this sense are not quantum channels. This is one respect in
which ``tripartite'' language is more flexible than ``channel'' language.

It is also worth observing that the superoperator $\EC$ uniquely determines
$\ket{\Om}$ up to local unitaries on $\HC_a$ and $\HC_c$ for a fixed $d_c$.
This is because a set of Kraus operators is generated, \eqref{eqn14}, using an
orthonormal basis $\{\ket{c_l}\}$ of $\HC_c$, and one can invert the process
by writing $V=\sum_l\ket{c_l}K_l$, where of course the result depends on the
choice of basis $\{\ket{c_l}\}$.  Different orthonormal bases on $\HC_c$, as
is well-known, simply give rise to different collections of Kraus operators
which represent the same quantum channel or operation. In this sense a channel
completely determines its complementary channel for a fixed $d_c$, and vice
versa, up to local unitaries. However, different insights may emerge by
considering one rather than the other, or by thinking about the two together.

We say there exists a \emph{perfect quantum channel} from $a$ to $b$ when all
types of information about $a$ are perfectly present in $b$.  This by itself
implies that $\rho_a=I_a/d_a$ (see Theorem 3 in \cite{Gri05}), and thus $\EC$
in \eqref{eqn18} is trace-preserving. It obviously suffices to check that the
information associated with every orthonormal basis is present in $b$, but
there are also weaker conditions that ensure the presence of a perfect
quantum channel; e.g. see \cite{Gri07, ChristWinterIEEE2005} and the discussion in Sec.~\ref{sbct5.3}.

A more general relationship is possible between an isometry $V$ and a
tripartite pure state, by starting with  \eqref{eqn7}, the circuit in
Fig.~\ref{fgr1}, but assuming that $\ket{\Phi}$, while no longer fully
entangled, has full Schmidt rank:
\begin{equation}
  \ket{\Phi} = \sum_k\sqrt{\pi_k}\, \ket{a_k}\ot\ket{a'_k},
\label{eqn19}
\end{equation}
with $\pi_k>0$ for every $k$.   With 
$V$ an isometry of the form \eqref{eqn8}, $j$ replaced by $k$, and 
\begin{equation}
  \rho_a = \sum_k \pi_k\dya{a_k}
\label{eqn22}
\end{equation}
the partial trace of $\dya{\Phi}$ down to $\HC_a$, one has
\begin{equation}
  \rho_b = \Tr_{ac}(\dya{\Om}) = \EC(\rho_a),
\label{eqn20}
\end{equation}
where $\EC$ is the superoperator corresponding to $V$ through \eqref{eqn13}.
A similar result holds for the complementary $a$ to $c$ channel.
The ket $\ket{\Om}$ determines the projector $\Upsilon=VV\ad$ uniquely, but $V$
itself only up to a unitary transformation on $\HC_a$. Conversely, two
isometries $V$
and $\widetilde{V}$ giving rise to the same $\Upsilon$ can be used to generate the
same $\ket{\Om}$ by using two different entangled states $\ket{\Phi}$ and
$\ket{\widetilde{\Phi}}$.

The partially entangled $\ket{\Phi}$ \eqref{eqn19} is useful when addressing the following
question: Suppose an ensemble $\{p_j,\rho_j\}$ of states is sent through the
quantum channel $\EC$.  How can one relate the outputs $\EC(\rho_j)$ of the
channel to corresponding outcomes of a POVM measurement $P_a$ on the
tripartite state $\ket{\Om}$?  Suppose the density operator $\rho_a=\sum_j
p_j\rho_j$ for the ensemble is of the form \eqref{eqn22}, i.e., choose
$\ket{\Phi}$ in \eqref{eqn19} such that this is the case. Then define
$P_a$ through
\begin{equation}
  P\trp_{aj} = p_j W\rho_j W\ad,
\label{eqn24}
\end{equation}
where $\trp$ denotes the transpose in the basis $\{\ket{a_k}\}$, and 
\begin{equation}
  W = \sum_k(1/\sqrt{\pi_k})\dyad{a_k}{a'_k}.
\label{eqn25}
\end{equation}
It is straightforward to show that $P_{aj}$ is a positive operator with the
same rank as $\rho_j$ (since $W$ is nonsingular), and $\sum_jP_{aj} =
I_a$. The probability of outcome $j$ for the POVM is
$p_j$, and the corresponding conditional density operator is
\begin{equation}
  \rho_{bcj} = \Tr_a(P_{aj}\rho_{abc}) = V\rho_j V\ad,
\label{eqn26}
\end{equation}
with $\rho_{abc} = \dya{\Om}$. 
Tracing this down to $b$ yields $\EC(\rho_j)$, the outcome when $\rho_j$ is
sent through the channel.

The preceding discussion requires some fairly obvious changes if some of the
$\pi_k$ in \eqref{eqn19} are zero.  First, $\Upsilon=VV\ad$ is not uniquely
determined by $\ket{\Om}$, since the $\ket{s_k}$ in \eqref{eqn8} corresponding
to zero $\pi_k$ are unknown.  Second, the sum in \eqref{eqn25} must be
restricted to the $k$ with $\pi_k > 0$, whereas \eqref{eqn24} remains the same.

\xb
\section{Information measures}
\label{sct3}
\xa

\xb
\subsection{Entropies}
\label{sbct3.1}
\xa

All the information measures that we will introduce are based on some sort of
entropy.  In classical information theory \cite{CvTh06} the usual
starting point is the Shannon entropy
\begin{equation}
  H(P)= H(\{p_j\}) = -\sum_j p_j\log p_j,
\label{eqn27}
\end{equation}
where $P$ denotes a random variable or its corresponding probability
distribution. Given two random variables $P$ and $Q$ the entropy $H(P,Q)$ is
obtained by replacing $p_j$ in \eqref{eqn27} by the joint probability
distribution $p_{jk}=\Pr(P_j,Q_k)$ and summing over both $j$ and $k$.  The
\emph{conditional entropy} and \emph{mutual information} are then defined by:
\begin{align}
  H(P|Q) &= H(P,Q) - H(Q),\nonumber\\
   H(P\colo Q) &= H(P)+H(Q) - H(P,Q).
\label{eqn28}
\end{align}

The quantum entropy most closely analogous to Shannon's $H$ is the von Neumann
entropy 
\begin{equation}
  S(\rho)=-\Tr (\rho \log \rho),
\label{eqn29}
\end{equation}
but we have also studied some other possibilities:
\begin{align}
S_R(\rho)&=\frac{1}{1-q} \log \Tr (\rho^q),
 \quad 0< q\leqslant 1,
\notag\\
S_T(\rho)&=\frac{1}{1-q} [\Tr (\rho^q)-1],
 \quad 0 < q\leqslant \infty,
\notag\\
S_Q(\rho)&=1-\Tr (\rho^2).
\label{eqn30}
\end{align} 
Here  $S_R$, $S_T$, and $S_Q$ are the \emph{Renyi, Tsallis},
and \emph{quadratic} (often misleadingly called \emph{linear}) entropies,
respectively.  Some of our results are valid for all these entropies, in which
case they are stated for $S_K$, where $K$ denotes either no subscript (von
Neumann) or else one of the three symbols $R,T,Q$.

All of these entropies are strictly concave, $S_K(\sum p_j\rho_j)\geqslant
\sum p_jS_K(\rho_j)$ for $0<p_j<1$ and $\sum p_j=1$, with equality if and only
if all $\rho_j$'s are equal, provided the parameter $q$ in the case of $S_R$
and $S_T$ is in range specified in \eqref{eqn30}.  Both $S_R$ and $S_T$ are
equal to $S$ in the limit $q=1$, and $S_T$ interpolates between $S$ and $S_Q$
as $q$ goes from 1 to 2.\footnote{For this remark (that $S_T=S$ for $q=1$) to be true, one should use base $e$ for the log appearing in $S$; however, we note that all other remarks and results in this article are valid for arbitrary base of the log.} The entropies $S$, $S_Q$, and $S_T$ for $q\geqslant
1$, are subadditive \cite{SubaddivityQ} in the sense that
$S_K(\rho_a)+S_K(\rho_b)\geqslant S_K(\rho_{ab})$, but only the von Neumann
$S$ has the property of strong subadditivity on a tripartite system (p. 519 of
\cite{NieChu00}):
\begin{equation}
\label{eqn31}
S(\rho_{ab})+S(\rho_{bc})\geqslant S(\rho_{abc})+S(\rho_b).
\end{equation}

Given a bipartite quantum system with a density operator $\rho_{ab}$, partial
traces $\rho_a$ and $\rho_b$, the \emph{quantum conditional entropy} and the
\emph{quantum mutual information} are defined as (p. 514 of \cite{NieChu00})
\begin{align}
  S(a|b) &= S(\rho_{ab}) - S(\rho_b),\nonumber\\ 
 S(a\colo b) &= S(\rho_a) + S(\rho_b)  - S(\rho_{ab}),  
\label{eqn32}
\end{align}
which are formally analogous to the quantities in \eqref{eqn28}. Note that $S(a|b)$ can be negative. 
On the other hand, $S(a\colo b)$ is nonnegative and vanishes for a product state
$\rho_{ab} = \rho_a\ot\rho_b$, and thus can be regarded in some sense as a measure of
how much information about $a$ is in $b$ or vice versa. Thought of in this way
it has the property that for a tripartite system $abc$,
\begin{equation}
\label{eqn33}
S(a\colo bc) \geqslant S(a\colo b),
\end{equation}
i.e., there is less information about $a$ in $b$, a subsystem of $bc$, than
there is in $bc$, which seems a reasonable requirement for a measure of
information. Note that \eqref{eqn33} is equivalent to \eqref{eqn31}, a property not shared (in general) by
the other entropies defined in \eqref{eqn30}.

We shall later prove our main result using the relative entropy:
\begin{equation}
\label{eqn33aa}
S(\rho||\sg)=\Tr(\rho\log \rho)-\Tr(\rho\log \sg),
\end{equation}
which has the useful property \cite{VedralReview02} that it is non-increasing under the action of a quantum channel $\EC$,
\begin{equation}
\label{eqn33bb}
S(\rho||\sg)\geq S(\EC(\rho)||\EC(\sg)).
\end{equation}
The extension of \eqref{eqn33aa} to general positive operators is natural, and \cite{OhPe93} for any positive operators $A$, $B$, and $C$, if $C\geq B$ (i.e.\ $C-B$ is a positive operator),
\begin{equation}
\label{eqn33cc}
S(A||B)\geq S(A||C).
\end{equation}

\xb
\subsection{Distinguishability measures}
\label{sbct3.2}
\xa

While $S(a\colo b)$ can serve as an overall indication of how much information
about $a$ is in $b$, or vice versa, it is not a measure that depends on the
\emph{type} of information, so cannot be used to compare how well different
types of information about $a$ are found in, or transmitted to $b$.  For this
purpose one could use a \emph{fidelity} measure: how closely a state on
$\HC_b$ resembles its counterpart on $\HC_a$. However, this requires making
some identification between the two Hilbert spaces, which is not easy to do if
they are of different dimension, or else one needs an additional map or
channel to carry $\HC_b$ back to $\HC_a$.  For this and other reasons we
prefer to use a \emph{distinguishability} measure.  Thus suppose
$P_a=\{P_{aj}\}$ is a decomposition of the identity $I_a$ of $\HC_a$,
\eqref{eqn1}, and $\{p_j,\rho_{bj}\}$ is the ensemble of conditional states on
$\HC_b$ defined in \eqref{eqn6}.  Two extreme cases were discussed in
Sec.~\ref{sbct2.1}: that in which the $P_a$ type of information is perfectly
present in $b$, which means $\rho_{bj}\rho_{bk}=0$ for $j\neq k$, thus
conditional density operators perfectly distinguishable; and the $P_a$ type of
information (completely) absent from $b$, meaning the $\rho_{bj}$ are
identical for all $j$ and thus indistinguishable. Our goal is to assign
numerical values to situations lying between these extremes.

Ideally one might use some collection of numbers referring to the
distinguishability of every pair of conditional density operators $\rho_{bj}$, see \cite{FuchsThesis} for an overview of distinguishability measures for two density operators. However, we shall employ a much coarser but still useful characterization in
which a single number, in some sense an ``average'' distinguishability, is
assigned to each information type, thereby allowing us to focus on our primary
goal: elucidating how the amount of information depends upon the type
considered, for a given pre-probability (density operator or channel). As is
customary in information theory we want a measure that is nonnegative, that
is (formally) invariant under local unitary operations, and, naturally, we prefer
simple mathematical expressions that have a clear intuitive
interpretation.  This still leaves many possibilities, but among them we have
found that measures based on the Holevo function 
\begin{equation}
\chi_K(\{p_j,\rho_j\}) = S_K(\sum_j p_j\rho_j)-\sum_j p_j S_K(\rho_j)
\label{eqn34}
\end{equation}
are particularly useful, where $\{p_j,\rho_j\}$ denotes an \emph{ensemble}
associated with a particular Hilbert space $\HC$: each $\rho_j$ a density
operator on this space, and the $\{p_j\}$ a probability distribution.  Here
$S_K$ could be any of the entropies defined in \eqref{eqn29} or
\eqref{eqn30}; $S$ without a subscript refers to the von Neumann entropy, and
the corresponding $\chi$ has no subscript.  Because each of these entropies is
a strictly concave function (for $q$ in the appropriate range indicated in
\eqref{eqn30}), $\chi_K$ is nonnegative and equal to zero if and only if the
$\rho_j$ are identical.

When \eqref{eqn34} is applied to the ensemble
$\{p_j,\rho_{bj}\}$ of \eqref{eqn6}, states in $\HC_b$ conditional on the
decomposition $P_a = \{P_{aj}\}$ in \eqref{eqn1}, the result is
\begin{equation}
\chi_K(P_a,b) := S_K(\rho_b)-\sum_j p_j S_K(\rho_{bj}),
\label{eqn35}
\end{equation}
a measure of the amount of information of type $P_a$ in $b$. This is also a numerical measure of what is sometimes called quantum side information \cite{RenesBoileau, DevWin03}.

While $P_a$ can refer to a general projective decomposition of $I_a$ or a POVM, we will often be interested in an orthonormal basis $\{\ket{w_j}\}$, projectors $\dya{w_j}$, of $\HC_a$, in which case we will write $\chi_K(w,b)$, omitting the $a$ subscript when it is obvious from the context. One can easily show using the concavity of $S_K$ that
\begin{equation}
\label{eqn36}
\chi_K(P_a,b) \geq \chi_K(\widetilde P_a,b),
\end{equation}
where $P_a$ and $\widetilde P_a$ are POVMs, and $\widetilde P_a$ is a coarse-graining of $P_a$ formed by summing some of the $P_{aj}$ elements. Also, as a consequence of \eqref{eqn31}, see \cite{SWW96}, 
\begin{equation}
\label{eqn37}
\chi(P_a,bc) \geq \chi(P_a,b),
\end{equation}
so a subsystem $b$ of $bc$ cannot contain more information than $bc$ itself. (This does not hold for $\chi_K$ with $K=R$, $T$ or $Q$.)

In the case of a quantum channel $\EC$ \eqref{eqn13} from $a$ to $b$ associated with isometry $V$ from $a$ to $bc$, we define
\begin{equation}
\label{eqn38}
\chi_K(P_a,\EC):=S_K[\EC(\sum p_j \rho_{aj})]-\sum p_j S_K[\EC(\rho_{aj})],
\end{equation}
where $P_a$ is a POVM, $I_a=\sum P_{aj}=d_a \sum p_j\rho_{aj}$, with 
\begin{equation}
\rho_{aj}=P_{aj}/\Tr(P_{aj}),\quad p_j=\Tr (P_{aj})/d_a.
\label{eqn39}
\end{equation}
Note that $\EC(\sum p_j \rho_{aj})=\Tr_c(VV\ad)/d_a=\Upsilon_b/d_a$ [see \eqref{eqn12}] in the first term of \eqref{eqn38} is independent of the POVM $P_a$. Equation \eqref{eqn38} is some measure for how well $\EC$ preserves the distinguishability of the $P_a$ ensemble; e.g. if $\EC$ perfectly preserves the orthogonality of an input orthonormal basis $w$ then $\chi(w,\EC)=\log d_a$, otherwise $\chi(w,\EC)< \log d_a$ (see Lemma~\ref{thm1} below). 

In contrast to $\chi(P_a,b)$, the quantity \cite{DevWin03}
\begin{equation}
H(P_a|b) := H(P_a) - \chi(P_a,b)
\label{eqn40}
\end{equation}
is a measure of \emph{absence} of the $P_a$ type of information from $b$, where $H(P_a)$ is the Shannon entropy \eqref{eqn27} associated with the probabilities defined in \eqref{eqn6}.\footnote{Following \cite{NieChu00}, we use $H$ for classical entropy and $S$ for quantum entropy. For conditional entropy, we use $H$ if the first argument is classical as in \eqref{eqn40}, and $S$ if the first argument is more general (quantum) as in \eqref{eqn32}.} One can also think of $H(P_a|b)$ as the \emph{missing information} about $P_a$ given the quantum system $b$, and it is a natural quantum analog of $H(P_a|Q_b)=H(P_a)-H(P_a\colo Q_b)$ [see \eqref{eqn28}], where one identifies $\chi(P_a,b)$ as a quantum analog of $H(P_a\colo Q_b)$.\footnote{It is straightforward to show that $\chi(P_a,b)$ becomes $H(P_a\colo Q_b)$ if one replaces the conditional density operators $\rho_{bj}$ in \eqref{eqn35} with conditional probability distributions $\Pr(Q_b|P_a=P_{aj})$, and also replaces $S()$ with $H()$.} In contrast to $S(a|b)$ defined in \eqref{eqn32}, $H(P_a|b)$ is non-negative (see Lemma~\ref{thm1}); it equals the Shannon missing information $H(P_a)$ in the case when $b$ provides no information about $P_a$, and it equals zero only when $b$ perfectly contains the $P_a$ information.

We remark that an alternative way of defining $H(P_a|b)$, similar to that employed in \cite{RenesBoileau, BertaEtAl, TomRen2010}, is to introduce the quantum channel $\EC_P$ from $ab\to eb$ defined by
\begin{equation}
\label{eqn45aaa} 
\EC_P(\rho_{ab})= \sum_j \dya{e_j}\ot\Tr_a(P_{aj}\rho_{ab}),
\end{equation}
where $\{\ket{e_j}\}$ is an orthonormal basis for an auxiliary system $e$. Then $H(P_a|b)$ is the von Neumann conditional entropy $S(e|b)$ of the state $\EC_P(\rho_{ab})$.

\begin{lemma}
\label{thm1}
This lemma summarizes some useful properties of the $\chi(P_a,b)$ and $H(P_a|b)$
measures.

(i) For any ensemble $\{p_j,\rho_j\}$
\begin{equation}
\label{eqn41}
\chi(\{p_j,\rho_j\})=S(\sum_j p_j\rho_j)-\sum_j p_jS(\rho_j)\leq H(\{p_j\}),
\end{equation}
with equality if and only if the $\rho_j$ are mutually orthogonal. 

(ii) Let $P_a$ and $Q_b$ be any two POVMs on $a$ and $b$ respectively,
  and $\rho_{ab}$ any state on $\HC_{ab}$.  Then
\begin{align}
\label{eqn42}
H(P_a \colo Q_b) &\leq \chi(P_a,b) \leq \nonumber \\
& \min\{S(\rho_a),S(\rho_b),S(a\colo b)\},
\end{align}
and hence by \eqref{eqn40}, \eqref{eqn41}, and \eqref{eqn42},
\begin{equation}
\label{eqn43}
0\leq H(P_a|b)\leq H(P_a|Q_b).
\end{equation}
 \openbox
\end{lemma}

Part (i) is from \cite{NieChu00} (Theorem~11.10, p.~518). The
left-hand-side of \eqref{eqn42} is Holevo's bound (p. 531 of \cite{NieChu00}),
and the right-hand-side of \eqref{eqn42} is similar to Proposition~1 of
\cite{WuEtAl2009} though we prove it in Appendix~\ref{aps2.1} since we have
explicitly inserted the bound on $\chi$.

\xb
\subsection{Entropy biases and coherent information}
\label{sbct3.3}
\xa

In addition to quantitative measures of information about one system present
in another it is useful to have measures of \emph{information differences}.
In what follows we shall make use of two quantities of this type.  When
considering two systems $b$ and $c$, 
\begin{equation}
  \Dl S_K(b,c) := S^{}_K(\rho_b)- S^{}_K(\rho_c),
\label{eqn46}
\end{equation}
is the \emph{entropy bias}, while for information type $P_a$,
\begin{equation}
\Dl \chi_K(P_a;b,c):=\chi_K(P_a,b)-\chi_K(P_a,c)
\label{eqn47}
\end{equation}
is the \emph{information bias}. Analogous quantities for the complementary channels $\EC$ and $\FC$ (to $b$ and $c$ respectively) arising from isometry $V$ are:
\begin{align}
\Dl S_K(\EC,\FC) &:= S^{}_K(\Upsilon_b/d_a)- S^{}_K(\Upsilon_c/d_a),\nonumber\\
\Dl \chi_K(P_a;\EC,\FC) &:= \chi_K(P_a,\EC)- \chi_K(P_a,\FC).
\label{eqn48}
\end{align}
Unlike our information measures these quantities can (obviously) be negative. When using the von Neumann entropy we omit the subscript $K$ and denote these quantities, e.g., by $\Dl S(b,c)$ and $\Dl \chi(P_a;b,c)$.

The \emph{coherent information} $I_{\text{coh}}$ (Sec.~12.4.2 of \cite{NieChu00}) is a particular instance of the entropy bias for the tripartite pure state $\ket{\Om}$:
\begin{equation}
I_{\text{coh}} (\rho_{a'},\EC) = \Dl S(b,c) 
\label{eqn49}
\end{equation}
where, see the discussion in Sec.~\ref{sbct2.2} associated with \eqref{eqn19}, the quantum channel $\EC$ corresponds to an isometry $V$ which yields $\ket{\Om}$ when applied to an entangled state $\ket{\Phi}$ chosen so that the partial trace of $\dya{\Phi}$ down to $a'$ yields the density operator $\rho_{a'}$.  The density operators $\rho_b$ and $\rho_c$ needed to define the entropy bias, \eqref{eqn46}, on the right side of \eqref{eqn49} are the partial traces of $\dya{\Om}$ down to systems $b$ and $c$, respectively. It can also be seen more directly, for the maximally-mixed input state, that $I_{\text{coh}}(I_{a'}/d_{a'},\EC) = \Dl S(\EC,\FC)$. 

Despite the connection in \eqref{eqn49}, the entropy bias in \eqref{eqn46} seems more natural in the state or static point of view, which lacks the notion of inputs and outputs, than $I_{\text{coh}}$. The latter has always been thought of as a function of a trace-preserving superoperator $\EC$ and an input state $\rho_{a'}$ to a channel, whereas the biases in \eqref{eqn46} and \eqref{eqn47} are simply functions of the tripartite state $\ket{\Om}$, without making reference to how it may have been generated by the combination of an isometry and a partially-entangled state.

\xb
\section{Basis Invariance}
\label{sct4}
\xa

We begin our discussion of how the amount of information about system $a$ in
some other system(s) depends on the type of information with two cases in
which certain quantities are actually independent of type.  In both of them, a pure-state
pre-probability is assumed. 

\begin{theorem}
\label{thm2}
Consider a bipartite system with a pure-state pre-probability $\rho_{ab}=\dya{\Psi}$. Let $N$ be a rank-1 POVM on $a$, let $w$ be an orthonormal basis (thus also a rank-1 POVM) on $a$, then
\begin{equation} 
\label{eqn50}
\chi_K(w,b)= \chi_K(N,b) =S_K(\rho_a)
\end{equation}
is \emph{independent} of the basis $w$ or rank-1 POVM $N$.
\end{theorem}

\begin{proof}
Apply \eqref{eqn35} to $w$, setting $S_K(\rho_b)=S_K(\rho_a)$ and the second term in \eqref{eqn35} to zero because each $\rho_{bj}$ is a pure state, proving $\chi_K(w,b)=S_K(\rho_a)$. From Sec.~\ref{sbct2.1}, $N$ is equivalent to an orthonormal basis $v_{ae}$ on $\HC_a\ot\HC_e$ assuming the state on $ae$ is $\rho_{ae}=\rho_a\ot\dya{e_0}$, where $\ket{e_0}$ is some pure state on $e$. Thus, $\chi_K(N,b)=\chi_K(v_{ae},b)=S_K(\rho_{ae})$, but $S_K(\rho_{ae})=S_K(\rho_a)$ for all entropy functions under consideration.
\end{proof}

This implies that if the $w$ information about $a$ is absent from $b$, $\chi_K(w,b)=0$, \emph{all} types are absent and $\ket{\Psi}$ is a product state, which is one form of the Absence theorem of \cite{Gri07}.  And it generalizes in that if the $w$ information is \emph{almost} absent from $b$, then by \eqref{eqn36} $\chi(w,b)\geq \chi(P,b)$, any other type $P$ is \emph{almost} absent from $b$. On the other hand, one can read \eqref{eqn50} as a statement that all (rank-1) types of information are equally present; the only problem is interpreting the common value of $\chi_K(w,b)=S_K(\rho_a)$. In the case of the von Neumann entropy, $\chi(w,b)=S(\rho_a)$ is the usual entanglement measure of $\ket{\Psi}$, and is an upper bound on the Shannon mutual information (Lemma~\ref{thm1}) that can be achieved by performing measurements in the Schmidt bases on $a$ and $b$. Note that reading \eqref{eqn50} in reverse provides a natural interpretation for $S_K(\rho_a)$; it is the amount of information about any rank-1 type $N$ contained in a system $b$ that purifies $\rho_a$, as measured by $\chi_K(N,b)$.

The following useful result for tripartite pure states and complementary channels (see Sec.~\ref{sbct6.1}) is proved in Appendix \ref{aps2.2}.
\begin{theorem}
\label{thm3}
Let $M$ and $N$ be rank-1 POVMs on $a$, and let $v$ and $w$ be orthonormal bases (thus also rank-1 POVMs) on $a$.

(i) Consider a tripartite system with a pure-state pre-probability $\rho_{abc}=\dya{\Om}$.  Then the information bias defined in 
\eqref{eqn47},
\begin{equation}
\Dl \chi_K(w;b,c)=\Dl \chi_K(N;b,c)=\Dl S_K(b,c),
\label{eqn51}
\end{equation}
where $K$ denotes any of the entropies defined in \eqref{eqn29} or
\eqref{eqn30}, is equal to the corresponding entropy bias, and thus
\emph{independent} of the choice of orthonormal basis or rank-1 POVM. It follows that the difference:
\begin{equation}
\chi_K(M,b)-\chi_K(N,b)=\chi_K(M,c)-\chi_K(N,c),
\label{eqn52}
\end{equation}
is the same for $b$ and $c$, which obviously holds if $M$ and $N$ are replaced by $v$ and $w$. 

(ii) Likewise, for complementary quantum channels $\EC$ and $\FC$, the information bias defined in \eqref{eqn48},
\begin{align}
\label{eqn53}
\Dl \chi_K(w;\EC,\FC)=\Dl \chi_K(N;\EC,\FC)=\Dl S_K(\EC,\FC)
\end{align}
is invariant to the choice of orthonormal basis $w$ or rank-1 POVM $N$, and
\begin{equation}
\label{eqn54}
 \chi_K(M,\EC)-\chi_K(N,\EC)= \chi_K(M,\FC)-\chi_K(N,\FC).
\end{equation} \openbox
\end{theorem}

This theorem provides a natural interpretation for the entropy bias of
a tripartite pure state: this is the amount by which more (or less if the bias
is negative) $w$ information about $a$ is present in $b$ than it
is in $c$.  The theorem tells us that this excess, which we call the
information bias, does not depend upon the orthonormal basis $w$, allowing us
to drop the $w$ from $\Dl \chi_K(b,c)$ under these conditions. This theorem is used in proving several of the results that
follow, including Theorems \ref{thm8}, \ref{thm10}, and \ref{thm11}.\\

\begin{example} 
\label{ex1}
As an illustration, suppose that in the case of a qubit, $d_a=2$, the $z$
information associated with the standard ${\ket{0},\ket{1}}$ basis is
perfectly transmitted from $a$ to $b$, while no information in the conjugate
$x$ basis is transmitted; i.e., we have a perfect ``classical'' channel from
$a$ to $b$. Setting $M=z$ and $N=x$ in \eqref{eqn52} and using Lemma \ref{thm1},
$H(z)=\chi(z,c)-\chi(x,c)$, which can only be true if $\chi(x,c)=0$ and
$H(z)=\chi(z,c)$. The $z$ information is thus perfectly transmitted from $a$
to $c$, saying the ``classical'' information (in this sense) is always copied
to another party, and further by the basis-invariance of $\Dl \chi(b,c)=0$, that
the $ab$ and the $ac$ channels are equally effective in terms of the $\chi$
measure.\footnote{An explicit example of this is the GHZ state $(\ket{000}+\ket{111})/\sqrt{2}$.} This conclusion can be reached by alternative lines of argument, but it illustrates the nontrivial content of Theorem~\ref{thm3}.
\end{example}

\xb
\section{Generalizing all-or-nothing theorems}
\label{sct5}
\xa

In this section we consider various quantitative generalizations, using the information measures introduced in Sec.~\ref{sct3}, of some ``all-or-nothing'' theorems \cite{Gri07}, which have the general form that in a multipartite system if a particular type or types of information about a
particular subsystem $a$ is perfectly present or absent in some other subsystem, then some other types of information about $a$ will also be
perfectly present or absent in other locations.  In each subsection below we provide a quantitative generalization of such a theorem to situations of
partial presence or absence, indicating the connection with the all-or-nothing theorem if it is not already clear.

\xb
\subsection{Truncation}
\label{sbct5.0}
\xa

The Truncation theorem of \cite{Gri07} states that if $\Pi=\{\Pi_j\}$ is a projective decomposition of $I_a$, and if the $\Pi$ type of information about $a$ is perfectly present in $c$, then for any third system $b$, the density operator $\rho_{ab}$ is truncated or block-diagonal (or ``pinched", p.\ 50 of \cite{Bhatia}) in the sense that $\rho_{ab}=\sum_j \Pi_{j}\rho_{ab}\Pi_{j}$. The following result is a generalization of this theorem to the case of partial information presence in $c$, and also allows for more general POVMs $P$ in Part (ii). The all-or-nothing result comes out by setting $H(\Pi|c)=0$ (perfect information presence) in \eqref{eqn1aaaa} below, which implies that $\rho_{ab}=\sum_j \Pi_{j}\rho_{ab}\Pi_{j}$ since $S(\rho||\sg)=0$ only if $\rho=\sg$. More generally, $\rho_{ab}$ will be ``close" (in the relative entropy sense) to the truncated form if $H(\Pi|c)$ is small.
\begin{lemma}
\label{thm4}
Let $\Pi=\{\Pi_{j}\}$ be a projective decomposition of $I_a$ and let $P=\{P_{j}\}$ be any POVM on $a$.

(i) Let $\rho_{abc}$ be a pure state, then
\begin{equation}
\label{eqn1aaaa}
H(\Pi|c)=S(\rho_{ab}||\sum_j \Pi_{j}\rho_{ab}\Pi_{j}).
\end{equation} 

(ii) Let $\rho_{abc}$ be \emph{any} state, then
\begin{equation}
\label{eqn1bbbb}
H(P|c)\geq S(\rho_{ab}||\sum_j P_{j}\rho_{ab}P_{j}).
\end{equation} \openbox
\end{lemma}

The proof can be found in Appendix~\ref{aps2.3}. This lemma is also used in proving the uncertainty relation in the next section.

\xb
\subsection{Information exclusion relations}
\label{sbct5.1}
\xa

An exclusion relation refers to incompatible types of information such that the presence of one type in one subsystem ``hinders'' or to some extent ``excludes'' the incompatible type from being present in a \emph{different} subsystem. Thus the Exclusion theorem of \cite{Gri07} asserts that if $v$ and $w$ are mutually unbiased bases on $a$, and the $v$ information about $a$ is perfectly present in $b$, then the $w$ information about $a$ is (completely) absent from $c$.  A quantitative extension of this to partial presence and absence can be based on the following theorem, where the incompatibility of two POVMs $P=\{P_j\}$ and $Q=\{Q_k\}$ is quantified using:
\begin{equation}
\label{eqn56}
r(P,Q):= \max_{j,k} \left\| \sqrt{P_j} \sqrt{Q_k}\right\|_{\infty}^2. 
\end{equation}
Here $\|\cdot\|_\infty$ denotes the supremum norm: the maximum singular value of the operator. 

Our main result, with proof in Appendix~\ref{aps2.4}, is:

\begin{theorem}
\label{thm5}
Let $\rho_{abc}$ be any state on $\HC_{abc}$.

(i) Let $P=\{ P_j \}$ and $Q=\{Q_k \}$ be any two POVMs on $\HC_a$, with $H(\cdot |\cdot)$ defined in \eqref{eqn40} and $r$ in \eqref{eqn56}. Then
\begin{equation}
H(P|b) +H(Q|c) \geq \log [1/r(P,Q)], 
\label{eqn57}
\end{equation}
where each $H(\cdot |\cdot)$ term is bounded by, e.g.:
\begin{equation}
H(P|b) \geq \log [1/\sqrt{r(P, P)}].
\label{eqn58}
\end{equation}

(ii) Specializing \eqref{eqn57} to the case of orthonormal bases $v=\{\dya{v_j}\}$ and $w=\{\dya{w_k}\}$, we obtain:
\begin{equation}
\label{eqn59}
H(v|b)+ H(w|c) \geq \log [1/r(v,w)],
\end{equation}
where in this case \eqref{eqn56} reads
\begin{equation}
r(v,w)=\max_{j,k} \vert\ip{v_j}{w_k}\vert^2.
\label{eqn60}
\end{equation}

(iii) The right-hand-side of \eqref{eqn59} is largest when $v$ and $w$ are MUBs, $r(v,w)=1/d_a$:
\begin{equation}
\label{eqn61}
H(v|b)+ H(w|c) \geq \log d_a.
\end{equation} \openbox
\end{theorem}

We remark that \eqref{eqn59} is equivalent to the main inequality conjectured in \cite{RenesBoileau} and proven in \cite{BertaEtAl}, see Sec.~\ref{sbct6.2}, and \eqref{eqn57} was also recently proven in \cite{TomRen2010} using smooth entropies, an approach different from ours. Our proof approach is based on the relative entropy; we will go into more detail about this approach in a subsequent article \cite{ColesEtAlTBP}. 

It is useful to view the inequalities in Theorem~\ref{thm5} in two different ways, as information exclusion relations \emph{and} as entropic uncertainty relations. The fact that they contain both principles can be seen, for example, in the MUB case by rewriting \eqref{eqn61} as:
\begin{equation}
\label{eqn62}
H(v)+H(w)\geqslant \chi(v,b)+\chi(w,c)+\log d_a.
\end{equation}
Viewed from the left-hand-side it looks like an entropic uncertainty relation:
a lower bound on an entropic sum. Viewed from the right-hand-side it looks
like an information exclusion relation: an upper bound on an information
sum. We note here that setting $H(v|b)=0$ in \eqref{eqn61} implies $H(w|c)=\log d_a$, the maximum value, and thus $c$ contains no information about $w$, demonstrating that our result implies (and thus generalizes) the Exclusion theorem from \cite{Gri07}.

As \eqref{eqn59} was proven in \cite{BertaEtAl}, consider the following example illustrating how \eqref{eqn57} goes beyond \eqref{eqn59}.
\begin{example} 
\label{ex2} Set $Q$ to the $w$ basis, and let $P$ be a POVM composed of $n$ pure states or rank-1 operators each with trace $d_a/n$ and each of which is unbiased with respect to the $w$ basis. Applying \eqref{eqn57} gives
\begin{equation}
\label{eqn63}
H(P|b)+H(w|c)\geq \log n.
\end{equation}
Now suppose $c$ contains all the $w$ information, $H(w|c)=0$. This implies that $H(P|b)=\log n$, which in turn implies \textit{two} conditions, the
probabilities of the $P_{j}$ are equal, so there is maximal missing information about which $P_{j}$ state system $a$ is in, \textit{and} the $P$ information must be perfectly absent from $b$: $\chi(P,b)=0$. The latter means that all states in $P$ get mapped by \eqref{eqn6} to the same output density operator $\rho_{bj}$ on $b$.
For example, for $d_a=2$ consider setting $w$ to the $z$ basis (standard basis); then $P$ could be the four states making up the $x$ and $y$ bases or three states forming an equilateral triangle in the $xy$ plane of the Bloch sphere or any symmetric set of states in the $xy$ plane. Imagining $P$ to be composed of a very large number of states in the $xy$ plane, by continuity \emph{all} states in the $xy$ plane must get mapped to the same output density operator $\rho_{bj}$ on $b$ when $H(z|c)=0$; a result that does not come out of pairing $z$ with a particular MUB, say $x$, and using \eqref{eqn59}. This all-or-nothing result is implied by the Truncation theorem of \cite{Gri07}, but \eqref{eqn63} also describes the partial information case, saying that the $\rho_{bj}$ associated with $P$ must be fairly indistinguishable if $H(w|c)$ is small.\end{example}

Inspired by (and strengthening) a result in \cite{KrishnaParth}, Eq. \eqref{eqn58} is, in some sense, an uncertainty relation for a \emph{single} POVM. Rewriting it as
\begin{equation}
\label{eqn64}
H(P)\geq \chi(P,b)+\log [1/\sqrt{r(P, P)}],
\end{equation}
it strengthens the bound $H(P)\geq \chi(P,b)$ from Lemma~\ref{thm1}; stating that if the $P$ measurement outcome is fairly certain [$H(P)$ small], this can partially exclude the $P$ information from another system $b$ [$\chi(P,b)$ small]. The idea is that a POVM is generally not a set of mutually-exclusive properties (Sec. \ref{sbct2.1}) so it has some intrinsic incompatibility, as measured by $\log [1/\sqrt{r(P, P)}]$. For example, if $P$ is composed of $n$ rank-1 operators each with trace $d_a/n$, then $\log [1/\sqrt{r(P, P)}]=\log(n/d_a)$.

Some information exclusion relations below for quantum channels are proven in Appendix~\ref{aps2.5}. Although they follow from Theorem~\ref{thm5}, they bring to mind a slightly different picture \cite{ChristWinterIEEE2005}, as one imagines Alice sending ``incompatible ensembles" $P$ and $Q$ respectively through $\EC$ and $\FC$, and if the $\FC$ channel transmits the $Q$ ensemble well to Carol, then the $\EC$ channel must be constructed in such a way that at its output Bob will have difficultly discerning which member of the $P$ ensemble Alice sends.
\begin{corollary}
\label{thm6}
For complementary quantum channels $\EC$ and $\FC$, $\chi$ given by \eqref{eqn38},

(i) Let $P$ and $Q$ be any two POVMs, with $H(P)=H(\{p_j\})$ where $p_j$ is given by \eqref{eqn39} and likewise for $H(Q)$,
\begin{align}
\label{eqn65}\chi(P,\EC) &\leq H(P)-\log[1/\sqrt{r(P,P)}],\\
\label{eqn66}\chi(P,\EC) + \chi(Q,\FC) &\leq H(P)+H(Q)-\log[1/r(P,Q)].
\end{align} 

(ii) For orthonormal bases $v$ and $w$,
\begin{equation}
\chi(v,\EC) + \chi(w,\FC) \leq \log [d_a^2 r(v,w)].
\label{eqn67}
\end{equation}

(iii) For MUBs $v$ and $w$,
\begin{equation}
\chi(v,\EC) + \chi(w,\FC) \leq \log d_a.
\label{eqn68}
\end{equation} \openbox
\end{corollary}

As another corollary to Theorem~\ref{thm5}, some uncertainty relations for a single system \cite{KrishnaParth,SanchezRuiz1995} (see Sec.~\ref{sbct6.2}) can be strengthened for mixed states, with the proof in Appendix~\ref{aps2.6}.
\begin{corollary}
\label{thm7}
$\mbox{       }$

(i) For any state $\rho$, let $N$ be a rank-1 POVM and let $P$ be any POVM, then
\begin{align}
\label{eqn69}
H(N) &\geq \log [1/\sqrt{r(N,N)}]+S(\rho),\\
\label{eqn70}
H(N)+H(P) &\geq \log [1/r(N,P)]+S(\rho).
\end{align}

(ii) For any state $\rho$ of a qubit (dimension $d=2$) and any complete set of three MUBs $x$, $y$, and $z$:
\begin{equation}
\label{eqn71}
H(x)+H(y)+H(z) \geq 2 \log 2 +S(\rho).
\end{equation} \openbox
\end{corollary}

While one might conjecture that \eqref{eqn69} or \eqref{eqn70} generalizes to the case where $N$ is an arbitrary POVM, it is easy to see that this is false. Imagine a highly mixed state such that $S(\rho)$ is very large, yet $N$ and $P$ are composed of coarse-grained projectors with very high rank, so $H(N)$ and $H(P)$ would be small, violating the inequality.

Note that \eqref{eqn71} is a tight bound, achieved for example when the state is along the $z$-axis of the Bloch sphere, such that $H(x)=H(y)=\log 2$ and $H(z)= S(\rho_a)$.

\xb
\subsection{Suppression of differences}\label{sbct5.2}
\xa

The following is a bipartite result, proved in Appendix~\ref{aps2.7}, saying that the presence of some type of information $P$ about $a$ in $b$ \emph{suppresses the difference} in the presence of two other types of information, $M$ and $N$, about $a$ in $b$. Note that a similar result holds for quantum channels.

\begin{theorem}
\label{thm8}
Let $\rho_{ab}$ be any state,

(i) For any POVM $P$ on $a$; rank-1 POVMs $M$ and $N$ on $a$,
\begin{align} 
\label{eqn72}
|&\chi(M,b)-\chi(N,b)| \leq H(P|b)+\nonumber\\
&\max \{H(M)-\log [1/r(P,M)],H(N)-\log [1/r(P,N)]\}, \nonumber \\
|&H(M|b)-H(N|b)|\leq H(P|b)+\nonumber \\
&\max \{H(M)-\log [1/r(P,N)],H(N)-\log [1/r(P,M)]\}.
\end{align}

(ii) For orthonormal bases $u$, $v$, $w$ on $a$, with $u$ and $v$ each MU with respect to $w$:
\begin{align} 
\label{eqn73}
|\chi(u,b)-\chi(v,b)| &\leq H(w|b),\nonumber \\
|H(u|b)-H(v|b)|&\leq H(w|b).
\end{align}

(iii) Let $u$, $v$, $w$ be as in (ii), but in addition assume that the $w$ type is perfectly present in $b$, then
\begin{align} 
\label{eqn74}
\chi_K(u,b)= \chi_K(v,b)&= S_K(\rho_b)-S_K(\rho_{ab}), \nonumber \\
H(u|b)=H(v|b)&=\log d_a+S(a|b),
\end{align}
meaning that all types MU to $w$ are present \emph{to the same degree} in $b$, in this sense. \openbox
\end{theorem}

The difference suppression effect is most apparent in part (ii) of this theorem, where \eqref{eqn73} says that the presence in $b$ of the $w$ information forces types MU to the $w$ type to be equally present in $b$, in the sense of having the same $\chi$ and $H$ quantities. As an illustration, consider $d_a=2$, let the $z$ information about $a$ be perfectly present in $b$, then all types in the $xy$ plane of the Bloch sphere are present in $b$ to the same degree, bringing to mind the image of a prolate spheroid (American football), with $z$ being the major axis, for the information about $a$ present in $b$.

\xb
\subsection{Decoupling theorems}\label{sbct5.3}
\xa

The preceding results can be used to generalize of some all-or-nothing decoupling theorems, which provide sufficient conditions about the information content of $b$ and/or $c$ to ensure that $c$ is completely uncorrelated to or decoupled from $a$. For example, the No Splitting theorem of \cite{Gri07} states that if all types of information about $a$ are perfectly present in $b$, then all types of information about $a$ are perfectly absent from $c$.  (The name is motivated by the idea that a perfect quantum channel from $a$ to $b$ allows no diversion or split off of information to a third location $c$.) In our notation this corresponds to the assertion that when $H(w|b)=0$ for every orthonormal basis $w$ of $\HC_a$, then $\chi(w,c)=0$ for every such basis. What follows is a quantitative generalization, a corollary of Theorem~\ref{thm5}, with the No Splitting theorem the special case when $\al=0$.
\begin{corollary}
\label{thm9}
Let $\al$ be some positive constant.

(i) For any state $\rho_{abc}$, if $H(w|b) \leq \al$ for every orthonormal basis $w$ of $\HC_a$, then $\chi(w,c)\leq \al$ for every such basis.

(ii) For complementary quantum channels $\EC$ and $\FC$, if $\chi(w,\EC)\geq \log d_a - \al$ for every orthonormal basis $w$ of the channel input, then $\chi(w,\FC)\leq \al$ for every such basis.
\end{corollary}
\begin{proof}
For any orthonormal basis $w$ there is a MU basis $v$, and thus \eqref{eqn61} implies that $H(w|c) \geq \log d_a -\al$ and hence, because $H(w)$ cannot exceed $\log d_a$, $\chi(w,c)$ cannot be greater than $\al$. The channel version follows by the same argument using \eqref{eqn68}.
\end{proof}

The Presence theorem of \cite{Gri07} states that if two \emph{strongly incompatible} types of information about $a$ are perfectly present in $b$,
then all information about $a$ is perfectly present in $b$, which is to say there is a perfect quantum channel from $a$ to $b$. Unfortunately, ``strongly
incompatible'' is a complicated concept, and it is not obvious how to extend it to a quantitative measure in the general case. Instead, we consider two POVMs $N$ and $P$, and the case where they are MUBs implies they are strongly incompatible types of information. The following theorem, proved in Appendix~\ref{aps2.8}, combines the notions of ``presence'' and ``no splitting'', and gradually specializes from POVMs to orthonormal bases to MUBs. Note that part (ii) of this theorem is stated for channels to remind the reader that each of our results for states has some analogous formulation for channels.\footnote{When one obtains an upper bound on a $\chi$ quantity for a channel $\FC$, as in Theorems \ref{thm6}, \ref{thm9}, \ref{thm10}, and \ref{thm11}, this bound also holds if one composes \emph{any} channel $\GC$ with $\FC$, i.e.\ $\chi(w,\GC\circ \FC)\leq \chi(w,\FC)$ by \eqref{eqn37}, which is useful if one is interested in bounding information in a \emph{subsystem} of the output of $\FC$.} 
\begin{theorem}
\label{thm10}
For any POVM $P$ on $a$; rank-1 POVMs $M$ and $N$ on $a$; orthonormal bases $u$, $v$, $w$ on $a$;

(i) For any bipartite state $\rho_{ab}$
\begin{equation}
\label{eqn75}H(N|b)+H(P|b)\geq \log [1/r(N,P)]+S(a|b),\\
\end{equation}
where, for any tripartite state $\rho_{abc}$, 
 \begin{equation}
\label{eqn76}
\chi(M,b)\geq [-S(a|b)], \text{    and     } H(M|c)\geq [-S(a|b)].
\end{equation}

(ii) For complementary quantum channels $\EC$ and $\FC$,
\begin{align}
\label{eqn77}&\chi(u,\EC) \geq  \Dl \chi(\EC,\FC) \geq \chi(v,\EC)+ \chi(w,\EC)-\log [d_a^2r(v,w)],\\
\label{eqn78}&\chi(u,\FC) \leq  \log [d_a^3 r(v,w)]-[\chi(v,\EC)+ \chi(w,\EC)].
\end{align} 

(iii) For MUBs $v$ and $w$, 
\begin{align}
\label{eqn79}&S(a\colo b)\geq 2\log d_a-2[H(v|b)+H(w|b)],\\
\label{eqn80}&S(a\colo c)\leq H(v|b)+H(w|b).
\end{align} \openbox
\end{theorem}

Corollary~\ref{thm9} gave a condition to guarantee that no information is present in $c$, and it is that all information is present in $b$. But what part (iii) of Theorem~\ref{thm10} shows is that one need not check that every single type of information is present in $b$; rather, simply check that $b$ contains \emph{two types} that are MUBs and this will completely decouple $c$ from $a$ \cite{RenesBoileau}. Part (ii) emphasizes that the information in $c$ (transmitted by $\FC$) can be upper bounded and that in $b$ (transmitted by $\EC$) lower bounded even when the two bases are not MUBs. Part (i) generalizes this notion further to POVMs. By \eqref{eqn75} one can lower-bound $[-S(a|b)]$, some measure of the entanglement between $a$ and $b$, just by knowing that $b$ contains information about a rank-1 POVM on $a$ and an arbitrary POVM on $a$. By \eqref{eqn76} this serves to lower-bound both the $M$ information in $b$ and the $M$ information \emph{missing} from $c$, for \emph{any} rank-1 POVM $M$. The application of such a relation, specialized to orthonormal bases, to quantum cryptography was discussed previously in \cite{BertaEtAl}, and the generalization to POVMs might turn out to be useful.

There is a seemingly odd restriction in \eqref{eqn75} that either $N$ or $P$ must be composed of rank-1 elements. One might conjecture that \eqref{eqn75} holds for arbitrary POVMs, but this is false. One can see this by choosing $\rho_{ab}=\rho_a\ot\rho_b$ in which case \eqref{eqn75} reduces to \eqref{eqn70}. As discussed following Corollary~\ref{thm7}, \eqref{eqn70} could be violated dramatically if $S(\rho_a)$ was large but both $N$ and $P$ were composed of high-rank projectors.

The following decoupling theorem considers the situation where some type of information $w$ is both perfectly present in $b$ and absent from $c$. It shows that this simple condition strikingly is enough to \emph{completely decouple} $c$ from $a$, and furthermore, for pure states, it leads to the suppression of differences between \emph{all} types of information in $b$. The theorem gradually specializes from all states to pure states to channels, with the proof in Appendix~\ref{aps2.9}.

\begin{theorem}
\label{thm11}
Let $L$, $M$, $N$ be rank-1 POVMs and $P$ be any POVM on $a$; let $v$ and $w$ be orthonormal bases on $a$,

(i) Let $\rho_{abc}$ be any state, then
\begin{equation}
\label{eqn81}
S(a\colo c)\leq \chi(N,c)+H(N|b)-\log[1/\sqrt{r(N,N)}].
\end{equation}
If the $N$ type of information is perfectly present in $b$, then
\begin{equation}
\label{eqn82}
\chi_K(P,c)\leq \chi_K(N,c).
\end{equation}
If, in addition, the $N$ type of information is absent from $c$, then \emph{all} types of information about $a$ are absent from $c$, i.e.\ $a$ and $c$ are completely uncorrelated: $\rho_{ac}=\rho_a \ot \rho_c$.

(ii) In the special case of pure states $\rho_{abc}=\dya{\Om}$,
\begin{equation}
\label{eqn83}
|\chi(L,b)-\chi(M,b)| \leq \chi(N,c)+H(N|b)-\log[1/\sqrt{r(N,N)}]
\end{equation}
and thus, in the extreme case where the $N$ type of information is perfectly present in $b$ and absent from $c$, 
\begin{equation}
\label{eqn84}
\chi(L,b)=\chi(M,b)=S(\rho_a)
\end{equation}
is \emph{independent} of rank-1 POVM (or orthonormal basis).

(iii) For complementary channels $\EC$ and $\FC$,
\begin{align}
\label{eqn85}
\chi(v,\FC) &\leq \chi(w,\FC)+[\log d_a- \chi(w, \EC)],\nonumber \\
\chi(v,\EC) &\geq \chi(w, \EC)- \chi(w,\FC).
\end{align}
Thus, if the $w$ type of information is perfectly present in the $\EC$ channel and absent from the $\FC$ channel, the same is true for \emph{all} types of information. This is a necessary and sufficient condition for $\EC$ being a perfect quantum channel and $\FC$ being a completely noisy channel. \openbox
\end{theorem}

\xb
\section{Connection with other work}
\label{sct6}
\xa

\xb
\subsection{Difference of Holevo quantities}
\label{sbct6.1}
\xa

Schumacher and Westmoreland \cite{Schu98} remarked that a difference in $\chi$ quantities associated with sending an ensemble of pure states through complementary quantum channels depends only on the average density operator of the input ensemble. This situation is equivalent to the one considered in \eqref{eqn51} of Theorem~\ref{thm3}, where a rank-1 POVM $N$ acts on system $a$ of a tripartite pure state $\ket{\Om}$. The equivalence follows from the discussion in Sec.~\ref{sbct2.2}; decompose $\ket{\Om}$ into an isometry $V$ acting on half of a bipartite pure state $\ket{\Phi}$, as in \eqref{eqn7} and Fig.~\ref{fgr1}, then by the construction in \eqref{eqn24}, any pure-state ensemble at the input of $V$ can be produced by an appropriate choice of $N$ and $\ket{\Phi}$.
Despite this equivalence, the notion of basis invariance or invariance to the rank-1 POVM $N$ emerges naturally out of the state view, since the average density operator of the input ensemble to $V$ is unaffected by choice of $N$. If one is willing to restrict to inputting the maximally-mixed average density operator, then the basis-invariance emerges in the channel view as well, as in \eqref{eqn53}.

\xb
\subsection{Entropic uncertainty relations}
\label{sbct6.2}
\xa

Our inequalities are related to several entropic uncertainty relations in the literature (see \cite{EURreview1} for a recent review), which are translated below into our notation. Maassen and Uffink \cite{MaassenUffink} proved an entropic uncertainty relation for measurements in orthonormal bases $v$ and $w$ on system $a$ for any state $\rho_a$:
\begin{equation}
\label{eqn86}
H(v)+H(w) \geq \log[1/r(v,w)].
\end{equation}
Krishna and Parthasarathy \cite{KrishnaParth} generalized this to POVMs $P$ and $Q$,
\begin{equation}
\label{eqn87}
H(P)+H(Q) \geq \log[1/r(P,Q)],
\end{equation}
and also stated an uncertainty relation for a single POVM
\begin{equation}
\label{eqn88}
H(P) \geq \log[1/\sqrt{r(P, P)}].
\end{equation}
Hall \cite{Hall1} incorporated into \eqref{eqn86} the idea of ``classical" side information, i.e.\ information about the outcome of a POVM $X_e$ acting on a system $e$ that may be correlated to $a$:
\begin{equation}
\label{eqn89}
H(v)+H(w) \geq \log [1/r(v,w)]+H(v\colo X_e)+H(w\colo X_e).
\end{equation}
Considering $e$ to be a composite system $bc$ and $X_e=Q_b\ot R_c$ a composite POVM, it follows from \eqref{eqn89} that:
\begin{equation}
\label{eqn90}
H(v|Q_b)+H(w|R_c) \geq \log[1/r(v,w)],
\end{equation}
see the discussion in \cite{RenesBoileau} where \eqref{eqn90} was termed the weak complementary information tradeoff and was ascribed to Cerf et al.\ \cite{CerfEtAl}. 

The inequalities in Theorem~\ref{thm5}: \eqref{eqn57}, \eqref{eqn58}, and \eqref{eqn59} respectively strengthen \eqref{eqn87}, \eqref{eqn88}, and \eqref{eqn86} by allowing for quantum side information, for example, information about property $P$ contained in another quantum system $b$, as measured by $\chi(P,b)$. The presence of such $\chi$ quantities, reducing the left-hand-sides of the Theorem~\ref{thm5} inequalities, is precisely what strengthens these bounds.
Equation \eqref{eqn89} follows from \eqref{eqn86} [and thus \eqref{eqn59}] by an argument that can be found in \cite{Hall1,RenesBoileau}. Equation \eqref{eqn90} follows from \eqref{eqn59} using the Holevo bound \eqref{eqn43}, $H(v|Q_b)\geq H(v|b)$ and $H(w|R_c) \geq H(w|c)$.

Equation \eqref{eqn59} is precisely the ``strong complementary information tradeoff'' conjectured by Renes and Boileau \cite{RenesBoileau} and later proven by Berta et al.\ \cite{BertaEtAl}. It is straightforward to show that our definition of $H(v|b)$ in \eqref{eqn40} is equivalent the definition employed in \cite{RenesBoileau} and \cite{BertaEtAl}, see \eqref{eqn45aaa}.

The main inequality in Berta et al.\ \cite{BertaEtAl},
\begin{equation}
\label{eqn94}
H(v|b)+H(w|b)\geq \log [1/r(v,w)]+S(a|b)
\end{equation}
was formulated for orthonormal bases $v$ and $w$, and we generalized it to POVMs (with at least one POVM being rank-1) in \eqref{eqn75}. Also, \eqref{eqn94} is equivalent to \eqref{eqn59} as follows. Apply \eqref{eqn59} to a pure state $\rho_{abc}$ and use $H(w|b)=H(w|c)-\Dl \chi(b,c)$ with $S(a|b)=-\Dl \chi(b,c)$ to get \eqref{eqn94}. Conversely, starting from \eqref{eqn94}, follow the reverse process to prove \eqref{eqn59} for pure states $\rho_{abc}$, and then \eqref{eqn59} for mixed states follows from \eqref{eqn37}. Thus, since \eqref{eqn59} is generalized to \emph{two} arbitrary POVMs by \eqref{eqn57}, \eqref{eqn57} and \eqref{eqn75} provide two alternative generalizations of \eqref{eqn94}. To prove \eqref{eqn94}, Berta et al.\ first proved an uncertainty relation involving smooth minimum and maximum entropies, and then invoked a lemma that these entropies approach the desired von Neumann entropic quantities under an appropriate asymptotic limit. In contrast, our proof does not use smooth entropies, but invokes the monotonicity of the relative entropy under quantum operations, so the approaches are conceptually different. 

Christandl and Winter \cite{ChristWinterIEEE2005} derived an information exclusion relation for quantum channels, which can be rearranged and expressed in our notation to read:
\begin{equation}
\label{eqn95}
\chi(x,\EC)+\chi(z,\FC)\leq \log d_a,
\end{equation}
where $x$ and $z$ are orthonormal bases related to each other by the $d$-dimensional quantum Fourier transform, and $\EC$ and $\FC$ are complementary quantum channels. Equation~\eqref{eqn68} of Corollary~\ref{thm6} generalizes this to arbitrary MUBs, and \eqref{eqn66} further generalizes to input ensembles associated with POVMs.

Our results strengthen some uncertainty relations in the case of mixed states. In the special case where $N$ is a rank-1 POVM, \eqref{eqn69} and \eqref{eqn70} respectively strengthen \eqref{eqn88} and \eqref{eqn87} with the addition of the $S(\rho_a)$ term. S\'anchez-Ruiz \cite{SanchezRuiz1995} proved an entropic uncertainty relation for sets of $d_a+1$ MUBs, which when applied to qubits ($d_a=2$) gives:
\begin{equation}
\label{eqn96}
H(x)+H(y)+H(z) \geq 2 \log 2.
\end{equation}
Likewise this is strengthened for mixed states by \eqref{eqn71}. Bounds depending on the purity of $\rho_a$ were also given in \cite{SanchezRuiz1995}; in the qubit case these bounds are implied by \eqref{eqn71}.

\xb
\subsection{No Splitting and Decoupling}
\label{sbct6.3}
\xa

Kretschmann et al.\ \cite{Spekkens08} have studied the degree to which a channel is error-correctable using a diamond-norm measure, and showed that when a channel is nearly perfect (in this sense) its complementary channel transmits very little information, and vice versa. B\'eny and Oreshkov \cite{BenyOreshkov2010} formulated a similar theorem for complementary channels, but in a general, symmetric fashion, using a fidelity measure. Hayden and Winter \cite{HW2010} have studied the degree to which a channel preserves the distinguishability of input states, and formulated the tradeoff in geometry-preservation between complementary channels using a trace-distance measure.  Each of these formulations generalize the No Splitting principle (see Sec.~\ref{sbct5.3}), although their information measures are of a different nature from the one we employ, and the connection between our approach and theirs remains to be determined. Intuitively, the No Splitting theorem should also be related to the notion that entanglement is monogamous. Quantitative expressions of entanglement monogamy have been found in terms of the concurrence and the squashed entanglement \cite{HHHH09}; as these are ``global'' measures of correlation, their relation to our information-type-specific measure is not obvious.

Renes and Boileau \cite{RenesBoileau} formulated a decoupling theorem as a corollary to their conjectured uncertainty relation [Eq. \eqref{eqn59}], stating that if $b$ contains the information about two sufficiently incompatible orthonormal bases of $a$, then the coupling of $c$ to $a$ can be upper-bounded. This is quite similar to our Theorem~\ref{thm10}, which extends this notion to two sufficiently incompatible POVMs.

\xb
\section{Conclusions}\label{sct7}
\xa

\xb
\subsection{Summary}
\label{sbct7.1}
\xa

Since our technical results in Secs.~\ref{sct4} and \ref{sct5} involve a large
number of theorems, the following comments are intended to assist the reader
in seeing how they are related to one another and to the definitions given earlier
in Secs.~\ref{sct2} and \ref{sct3}.

In Sec.~\ref{sbct2.1} we generalize an earlier \cite{Gri07} notion of types of
quantum information to include general POVMs on a Hilbert space $\HC_a$ for
system $a$, by noting that the associated probabilities are the same as those
for a projective decomposition of the identity on a larger Hilbert space
$\HC_A$, the Naimark extension, and a rank-1 POVM corresponds to an
orthonormal basis of the extension.  Various measures for different types of
information are introduced and discussed in Sec.~\ref{sct3}.  For uniformity
of notation, Shannon entropies and related quantities are denoted by $H()$;
e.g., $H(P_a)$ is the missing information about type $P_a$, as determined by
its probability distribution, when the quantum state is assumed known. For
quantum entropies we use $S()$ for the von Neumann entropy, and $S_K()$, where
$K$ can be $R$ or $T$ or $Q$ for Renyi, Tsallis, and quadratic entropies,
respectively.

We use the Holevo function $\chi(P_a,b)$, or $\chi_K(P_a,b)$ for $S_K$,
\eqref{eqn35}, as a measure of the amount of information of type $P_a$ about
system $a$ which is present in system $b$, along with the complementary
quantity $H(P_a|b)$, \eqref{eqn40}, as a corresponding measure of the
amount of information about $P_a$ that is still missing given system $b$.
While the analogy is not exact, $\chi(P_a,b)$ is similar to Shannon's mutual
information $H(P_a\colo Q_b)$, whereas $H(P_a|b)$ resembles Shannon's
conditional entropy $H(P_a|Q_b)$.  In particular, $H(P_a|b)$, like
$H(P_a|Q_b)$, is nonnegative, so retains some of the intuition of the latter
quantity, in contrast to the quantum conditional entropy $S(a|b)$,
\eqref{eqn32}, which can be of either sign.  We use the term information bias
for the difference between the amount of type $P_a$ information about $a$ in
$b$ and the amount in $c$, $\chi(P_a,b)-\chi(P_a,c) = \Dl\chi(P_a;b,c)$, which
can have either sign. Similarly, we refer to $\Dl S(b,c) = S(\rho_b)-S(\rho_c)$ as the
entropy bias, and add a subscript $K$ when using an alternative to the von
Neumann entropy.
Our most extensive results are for the von Neumann entropy and its associated
information measures.  However, in some cases, see Theorems \ref{thm2},
\ref{thm3}, \ref{thm8}, and \ref{thm11}, these results also hold for a more
general $\chi_K$, and stating them in this form seems worthwhile, as for
certain purposes these other measures could be useful.

While the most natural and symmetrical, in terms of treating the different
parts on the same footing, formulation of our results is in terms of a
tripartite system, some of the more interesting and significant
applications are to quantum channels and complementary channels.  The
relationship between the tripartite and the channel perspectives is worked out
in some detail in Sec.~\ref{sbct2.2}, and in Sec.~\ref{sbct3.3} we relate the
coherent information for a quantum channel to a corresponding tripartite
entropy bias.  In several theorems the channel counterparts of tripartite
results are stated separately, because while the formal results are in some
sense the same, one's intuition about their significance can be different.

Our first set of results are the equalities in Theorems \ref{thm2} and
\ref{thm3} of Sec.~\ref{sct4}, which apply for pure quantum states of
bipartite and tripartite systems, respectively.  The first says that the
amount of information about $a$ in $b$ is independent of the type of
information, provided the latter is a rank-1 POVM; this includes an
orthonormal basis.  The second says that the difference between the amount of
information concerning such a rank-1 POVM in $b$ and in $c$ is independent of
the type considered, and equal to the corresponding entropy bias.
Equivalently, given two rank-1 POVMs $M$ and $N$, the difference between the
amount of $M$ and $N$ information about $a$ found in $b$ is the same as the
corresponding difference in $c$.  While these results are limited to pure
states, they are important for the proofs of many of the later results. They
also extend from von Neumann to other quantum entropies, so they are stated in
this more general form.

Perhaps the simplest way of viewing the collection of inequalities that make up Sec.~\ref{sct5} is that the main theorems are quantitative generalizations
of all-or-nothing theorems which can be stated quite concisely for types of information associated with orthonormal bases $v$ and $w$ of system $a$. A central result of this paper is Theorem~\ref{thm5}, and part (iii) of this theorem tells us that if the $v$ information about
$a$ is perfectly present in $b$, which is to say $H(v|b)=0$,
then the mutually unbiased (MU) $w$ type of information must be perfectly
absent from $c$: $H(w|c) = \log d_a$ means that $\chi(w,c)=0$. Part
(ii) allows for bases that are not MU at the cost of a weaker bound on
the $H$ measures, while part (i) is not restricted to bases but applies to
quite general types of information $P$ and $Q$.  The generalization to POVMs
is, in turn, based on Lemma~\ref{thm4}, which itself generalizes the Truncation theorem \cite{Gri07}: if the $v=\{v_j\}$ information is perfectly present in $c$ then $\rho_{ab}$ commutes with the $v_j$ projectors. 

The connections of Theorem~\ref{thm5} to literature entropic uncertainty relations are given in Sec.~\ref{sbct6.2}. Broadly speaking we think that the addition of quantum side information to uncertainty relations \cite{RenesBoileau,BertaEtAl} not only strengthens certain bounds but also gives further conceptual insight into the nature of complementarity, in that side information about complementary observables in different locations (Sec.~\ref{sbct5.1}) must be constrained as well. We also note a recent experimental study \cite{PrevedelEtAl2010}. Further remarks on the significance of Theorem~\ref{thm5} can be found in the discussion that follows it in Sec.~\ref{sbct5.1}.

Corollary~\ref{thm6} of Theorem~\ref{thm5} gives the corresponding result for quantum channels, generalizing to partial information and to arbitrary POVMs or orthonormal bases the all-or-nothing theorem: if the $v$ information is perfectly present in (or transmitted by) the $\EC$ channel, any MU type of information $w$ must be absent from (or destroyed by) the complementary channel $\FC$. In addition, Corollary~\ref{thm7} of Theorem~\ref{thm5} provides strengthened
information inequalities for a single system described by a mixed state.

The idea behind Theorem~\ref{thm8} is encapsulated
in the observation that if the information about an orthonormal basis $w$ of
$a$ is perfectly present in $b$, so $H(w|b)=0$, and $u$ and $v$ are bases of
$a$ that are MU with respect to $w$ (but not necessarily with respect to each
other) then the $u$ and $v$ types are present in $b$ in equal amounts.  If, on
the other hand the $w$ information is less than perfectly present in $b$, this
theorem provides quantitative bounds on the difference between the $u$ and $v$
types of information in $b$.  Similarly, the requirement that $u$ and $v$ be
MU relative to $w$ can be relaxed, and they can even be replaced with rank-1
POVMs, and $w$ with a general POVM, see part (i) of the theorem, at the price of
appropriately weakening the bounds that confine the differences.

The results in Sec.~\ref{sbct5.3} provide quantitative generalizations of conditions that ensure system $c$ is completely uncorrelated to (or decoupled from) system $a$, $\rho_{ac}=\rho_a\ot\rho_c$. Corollary~\ref{thm9} of Theorem~\ref{thm5} says that the correlations between $a$ and $c$ are tightly upper-bounded if system $b$ \emph{almost} perfectly contains all types of information about $a$, and gives the analogous result for complementary channels $\EC$ and $\FC$.
But Theorem~\ref{thm10} stresses the importance of the presence of \emph{just two} (sufficiently incompatible) types of information. That is, if $b$ perfectly contains the information about two MUBs of $a$, then $b$ contains \emph{all} types of information about $a$, and $c$ is completely uncorrelated to $a$; a generalization of this statement for the partial information case is given in part (iii) of Theorem~\ref{thm10}. Parts (ii) and (i) of this theorem respectively illustrate that this idea can be extended, at the price of weakened bounds, to any two orthonormal bases or to two POVMs in which at least one of the POVMs is rank-1. The relevance of inequalities like \eqref{eqn75} of Theorem~\ref{thm10}, where the presence of two types of information about $a$ in $b$ can be used to upper bound the information about $a$ in $c$, to quantum cryptography was discussed in \cite{BertaEtAl}. 

The same sort of decoupling occurs when a single type of information about $a$ associated with an orthonormal basis $w$ is perfectly present in
$b$ and completely absent from $c$.  Theorem~\ref{thm11} contains this interesting result together with certain quantitative generalizations, both
when the type of information in question is only partially absent from $c$, and when it is not perfectly present in $b$.

\xb
\subsection{Future outlook}
\label{sbct7.2}
\xa

There are various ways in which the results summarized above suggest problems
which deserve further attention and research.  One has to do with the
difference between rank-1 and higher-rank POVMs, or orthonormal bases as
against coarser projective decompositions of the identity.  In a number of
cases the results we have obtained for the former are distinctly stronger than
for the latter, but the reason for this is not always clear.  Since
applications of quantum information theory to macroscopic systems, in
particular to problems of decoherence, lead rather naturally to coarse
decompositions or POVMs, a good intuitive understanding in addition to formal
expressions would be of value. A second item concerns the use of the $r(P,Q)$ overlap measure for POVMs, or
its $r(v,w)$ counterpart for orthonormal bases, see \eqref{eqn56} and
\eqref{eqn60}. While this provides the basis of significant inequalities in
Theorem~\ref{thm5} and later, the fact that $r(P,Q)$ requires one to maximize
over all pairs of elements from the two POVMs hints that stronger results
might well be possible were one to use a more refined perspective on how the
POVMs are related to each other, or the sorts of information that they provide.

While qualitative inequalities are certainly an advance over simple
all-or-nothing results, it would be even better if one could express
information tradeoffs in terms of equalities of the sort which could
conceivably allow one to completely characterize how a quantum channel is
related to its complementary channel using a (hopefully small) number of
parameters with a clear intuitive significance.  The equalities in
Theorem~\ref{thm3}, as applied either to channels or, more generally,
pure-state tripartite systems, hint that something like this might be
possible, but thus far we have not found it.

Any advance in understanding tripartite systems raises an obvious question:
what about systems with four (or more) parts?  It is, of course, possible to study them by thinking of two of the parts as constituting a single object,
and then applying results for tripartite systems.  But there is probably some ``residual'' aspect of a system of four parts which cannot be captured in this way, just as
there are residual aspects of tripartite systems which cannot be understood simply in terms of combining two of them so as to yield a
bipartite system.  We think that our results in this paper have helped
to clarify some of this tripartite residual, and we hope they provide
hints on ways to deal with more complicated cases.

\begin{acknowledgments}
We thank Michael Zwolak, Luc Tartar, Edward Gerjuoy, Shiang-Yong Looi, and Danquynh Nguyen for helpful discussions. We are especially grateful to an anonymous referee for pointing out a serious error in an earlier version of this paper. The research described here was supported by the Office of Naval Research and by the National Science Foundation through Grant No. PHY-0757251. 
\end{acknowledgments}

\xb
\appendix

\xb
\section{Proof of Lemma~\ref{thm1}}\label{aps2.1}
\xa

\begin{proof}
(ii) The inequality $\chi(P_a,b)\leq S(\rho_b)$ obviously follows from \eqref{eqn35}. Now let $c$ be a system that purifies $\rho_{ab}$. Then by \eqref{eqn37}, $\chi(P_a,b)\leq \chi(P_a,bc)=S(\rho_a)-\sum_j p_j S(\rho_{bcj})\leq S(\rho_a)$. 

To prove $S(a\colo b)\geq \chi(P_a,b)$, as in Sec.~\ref{sbct2.1} think of $P_a$ as a \emph{projective} measurement $\widetilde w_{ae}$ on system $ae$, where $\widetilde w_{ae}$ is a coarse graining of some orthonormal basis (rank-1 projectors) $w_{ae}$. Let $c$ purify $\rho_{ab}$ such that $\rho_{abce}=\rho_{abc}\ot\dya{e_0}$ is a pure state. Then, $S(a\colo b)=S(\rho_{ae})+S(\rho_b)-S(\rho_c)=\chi(w_{ae},bc)+\chi(w_{ae},b)-\chi(w_{ae},c)\geq \chi(w_{ae},b)\geq \chi(\widetilde w_{ae},b)=\chi(P_a,b)$, by the Theorems in Sec.~\ref{sct4}, by \eqref{eqn37}, and by \eqref{eqn36}.
\end{proof}

\xb
\section{Proof of Theorem \ref{thm3}}\label{aps2.2}
\xa

\begin{proof}
(i) For orthonormal basis $w=\{\ket{w_j}\}$, insert \eqref{eqn35} into \eqref{eqn47} to obtain
\begin{align}
 &\chi_K(w,b)-\chi_K(w,c) \nonumber\\
 &= S_K(\rho_b) - S_K(\rho_c) - \sum_j p_j[S_K(\rho_{bj}) - S_K(\rho_{cj})].
\label{eqn99}
\end{align}
The final term vanishes, for the following reason.  
Write $\ket{\Om}=\sum_j\ket{w_j}\ot\ket{s_j}$ in the form \eqref{eqn10}
with $\ket{w_j}$ replacing $\ket{a_j}$, so from \eqref{eqn6} the conditional density operators
in \eqref{eqn99} are given by
\begin{equation}
 p_j\rho_{bj} = \Tr_c\Bigl(\dya{s_j}\Bigr),\quad 
 p_j\rho_{cj} = \Tr_b\Bigl(\dya{s_j}\Bigr).
\label{eqn100}
\end{equation}
Since $\ket{s_j}$ is a pure state the
partial traces
$\rho_{bj}$ and $\rho_{cj}$ have the same eigenvalues (determined by the Schmidt expansion coefficients of $\ket{s_j}$), except one may have more zeros than the other if $d_b\neq d_c$. Since $S_K(\rho)$ is a function only of the nonzero (positive) eigenvalues of $\rho$, each term in the final sum in \eqref{eqn99} vanishes, and we are left with \eqref{eqn51}. The generalization to rank-1 POVMs follows by the equivalence of $N$ to an orthonormal basis $v_{A}$ on $\HC_A$, the Naimark extension of $\HC_a$ as in Sec.~\ref{sbct2.1}. Since $\rho_{Abc}$ is a tripartite pure state, then $\Dl\chi_K(N;b,c)=\Dl\chi_K(v_{A};b,c)=S_K(\rho_b)-S_K(\rho_c)$.

(ii) Equation \eqref{eqn53} follows from \eqref{eqn51} by applying it to a channel ket $\ket{\Om}$ constructed from $V$ by \eqref{eqn7}. Alternatively, it can be proven directly from \eqref{eqn38} and \eqref{eqn48}, obtaining an equation similar to \eqref{eqn99},
\begin{align}
 \Dl \chi_K(P;\EC,\FC)=&S_K(\Upsilon_b/d_a)-S_K(\Upsilon_c/d_a) \nonumber\\
+& \sum_j p_j [S_K(\rho_{bj}) - S_K(\rho_{cj})],
\label{eqn101}
\end{align}
where the final term vanishes again since $\rho_{bj}=\Tr_c[V\rho_{aj}V\ad]$ and $\rho_{cj}=\Tr_b[V\rho_{aj}V\ad]$ have the same (non-zero) spectrum, as the $\rho_{aj}$ in \eqref{eqn39} are rank-1 operators.

Equations \eqref{eqn52} and \eqref{eqn54} follow immediately from \eqref{eqn51} and \eqref{eqn53}, respectively.
\end{proof}

\xb
\section{Proof of Lemma \ref{thm4}}\label{aps2.3}
\xa

\begin{proof}
(i)
\begin{align}
&S(\rho_{ab}||\sum_j \Pi_{j}\rho_{ab}\Pi_{j})\nonumber\\
\label{eqnNew60}&=-S(\rho_{ab})-\Tr[\rho_{ab}\log (\sum_j \Pi_{j}\rho_{ab}\Pi_{j})]\\
\label{eqnNew61}&= -S(\rho_c)-\Tr[\rho_{ab}\sum_k \Pi_{k}\log (\sum_j \Pi_{j}\rho_{ab}\Pi_{j})\sum_l\Pi_{l}]\\
&=-S(\rho_c)-\Tr[\sum_k \Pi_{k}\rho_{ab}  \Pi_{k}\log (\sum_j \Pi_{j}\rho_{ab}\Pi_{j})]\nonumber \\ 
\label{eqnNew62}&-\sum_{k,l\neq k} \Tr[\rho_{ab}  \Pi_{k}\log (\sum_j \Pi_{j}\rho_{ab}\Pi_{j})\Pi_{l}]\\
\label{eqnNew63}&=-S(\rho_c)+S(\sum_j \Pi_{j}\rho_{ab}  \Pi_{j})\\
\label{eqnNew64}&=-S(\rho_c)+H(\Pi)+\sum_j p_j S(\rho_{abj})\\
\label{eqnNew65}&=H(\Pi)-\chi(\Pi,c)= H(\Pi |c),
\end{align}
where $p_j=\Tr(\Pi_{j}\rho_{ab})$ and $p_j\rho_{abj}= \Pi_{j}\rho_{ab}\Pi_{j}$. The last term in \eqref{eqnNew62} disappears because $\log (\sum_j \Pi_{j}\rho_{ab}\Pi_{j}) $ is block diagonal with respect to the $\Pi_j$ projectors, and then one takes an off-diagonal element of it. Step~\eqref{eqnNew64} follows from Lemma~\ref{thm1}, part (i).

(ii) For clarity, we include the subscript $a$ on the POVM $P_a$. Think of $P_a=\{P_{aj}\}$ as a projective measurement $\Pi_A=\{\Pi_{Aj}\}$ on an extended Hilbert space $\HC_A$ (Naimark extension), with $\HC_a$ a subspace and $E_a$ the projector onto this subspace, and $P_{aj}=E_a\Pi_{Aj}E_a$. The state $\rho_{Ab}$ is the same as $\rho_{ab}$ but now just embedded in a larger space, that is: $\rho_{Ab}=E_a\rho_{Ab}E_a=\rho_{ab}$. Let $E^\bot_a$ be the projector onto the orthogonal complement of $\HC_a$, note $E^\bot_a\rho_{Ab}E^\bot_a=0$, and let the channel $\FC$ be defined by $\FC(\rho)= E_a \rho E_a +E^\bot_a \rho E^\bot_a $. Then if $d$ is a system that purifies $\rho_{abc}$, we have:
\begin{align}
\label{eqnNew70}&H(P_a|c)=H(\Pi_A|c)\\
\label{eqnNew71}&\geq H(\Pi_A|cd)=S(\rho_{Ab}||\sum_j \Pi_{Aj}\rho_{Ab}\Pi_{Aj})\\
\label{eqnNew72}&\geq S(\FC(\rho_{Ab})|| \FC(\sum_j \Pi_{Aj}\rho_{Ab}\Pi_{Aj}))\\
&= S(E_a\rho_{Ab} E_a || E_a(\sum_j \Pi_{Aj}\rho_{Ab}\Pi_{Aj}) E_a+\nonumber\\ 
\label{eqnNew73}&E^\bot_a(\sum_j \Pi_{Aj}\rho_{Ab}\Pi_{Aj}) E^\bot_a)\\
\label{eqnNew74}&= S(E_a\rho_{Ab} E_a || \sum_j E_a\Pi_{Aj}E_a\rho_{Ab}E_a\Pi_{Aj} E_a)\\
\label{eqnNew75}&= S(\rho_{ab} || \sum_j P_{aj}\rho_{ab}P_{aj}).
\end{align}
Note that the term with $E^\bot_a$ in \eqref{eqnNew73} disappeared because it lies outside of the support of $E_a\rho_{Ab} E_a$.
\end{proof}

\xb
\section{Proof of Theorem \ref{thm5}}\label{aps2.4}
\xa

\begin{proof}
First let us prove the single-POVM uncertainty relation as follows, defining $\lm_{\max}(\cdot)$ to be the maximum eigenvalue. From Lemma~\ref{thm4},
\begin{align}
\label{eqnNew1}H&(P|b)\geq S(\rho_{ac}||\sum_j P_{j}\rho_{ac}P_{j})\\
\label{eqnNew2}&\geq S(\rho_{c}||\sum_j \Tr_a[P_{j}\rho_{ac}P_{j}])\\
\label{eqnNew3}&\geq S(\rho_{c}||\sum_j \lm_{\max}(P_{j})\Tr_a[P_{j}\rho_{ac}])\\
\label{eqnNew4}&\geq S(\rho_{c}|| \max_j\lm_{\max}(P_{j}) \sum_j \Tr_a[P_{j}\rho_{ac}])\\
\label{eqnNew5}&= S(\rho_{c}|| \max_j\lm_{\max}(P_{j}) \rho_{c})\\
\label{eqnNew6}&= -\log \max_{j} \lm_{\max}(P_{j})= -\log \max_{j} \|P_{j}\|_\infty\\
\label{eqnNew7}& \geq -\log \max_{j,k} \|\sqrt{P_{j}}\sqrt{P_{k}}\|_\infty.
\end{align}
We invoked \eqref{eqn33bb} for step \eqref{eqnNew2}. We used \eqref{eqn33cc} for step \eqref{eqnNew3}, $\lm_{\max}(P_{j})I_a\geq P_j$ which implies $\Tr_a[\lm_{\max}(P_{j})I_aT_{ac}]\geq \Tr_a[P_jT_{ac}]$, where $T_{ac}=\sqrt{P_j}\rho_{ac}\sqrt{P_j}$ is a positive operator. We also used \eqref{eqn33cc} for step \eqref{eqnNew4}, $\max_j \lm_{\max}(P_{j}) \sum_j A_j \geq \sum_j \lm_{\max}(P_{j})A_j$ where the $A_j$ are positive operators.

Now for the two-POVM uncertainty relation, consider the quantum channel [as in \eqref{eqn45aaa}] $\EC_Q(\rho_{ab})=\sum_k \dya{e_k}\ot \Tr_a(Q_{k}\rho_{ab})$ associated with the $Q$ measurement, where $\{\ket{e_k}\}$ is an orthonormal basis of an auxiliary system $e$. One can verify that $\EC_Q$ is trace-preserving, and its complete positivity follows from the fact that $(\EC_Q \ot \IC_c)(\rho_{abc})= \sum_k \dya{e_k}\ot \Tr_a(Q_{k}\rho_{abc})$ is a positive operator for any system $c$, where $\IC_c$ is the identity channel for $c$. Also, define $G_{jk}= \sqrt{P_{j}} Q_{k} \sqrt{P_{j}} $, and note $G_{jk}\leq \lm_{\max}(G_{jk})I_a$, and $r(P,Q)=\max_{j,k}\lm_{\max}(G_{jk})$. Then from Lemma~\ref{thm4},

\begin{widetext}
\begin{align}
\label{eqnNew10}&H(P|c)\geq S(\rho_{ab}||\sum_j P_{j}\rho_{ab}P_{j})\geq S(\EC_Q(\rho_{ab})||\sum_j \EC_Q(P_{j}\rho_{ab}P_{j}))\\
\label{eqnNew14}&= S(\sum_{l} \dya{e_l}\ot \Tr_a\{Q_{l} \rho_{ab}\}||\sum_{j,k} \dya{e_k}\ot \Tr_a\{G_{jk}\sqrt{P_j}\rho_{ab}\sqrt{P_j}\})\\
\label{eqnNew15}&\geq S(\sum_{l} \dya{e_l}\ot \Tr_a\{Q_{l} \rho_{ab}\}|| \sum_{j,k} \lm_{\max}(G_{jk})\dya{e_k}\ot \Tr_a\{P_{j}\rho_{ab}\})\\
\label{eqnNew16}&\geq S(\sum_{l} \dya{e_l}\ot \Tr_a\{Q_{l} \rho_{ab}\}|| r(P,Q) I_e \ot \rho_{b})\\
\label{eqnNew17}&= -\log  r(P,Q) - S(\sum_{l} \dya{e_l}\ot \Tr_a\{Q_{l} \rho_{ab}\})- \Tr[(\sum_{l} \dya{e_l}\ot \Tr_a\{Q_{l} \rho_{ab}\})\log ( I_e \ot \rho_{b})]\\
\label{eqnNew18}&= -\log  r(P,Q) - H(Q)-\sum_l q_l S(\rho^{Q}_{bl}) +S(\rho_b)= -\log  r(P,Q) - H(Q|b),
\end{align}
\end{widetext}
where $q_l = \Tr(Q_{l} \rho_{ab})$ and $q_l \rho^{Q}_{bl}= \Tr_a(Q_{l} \rho_{ab})$. We invoked \eqref{eqn33bb} for step \eqref{eqnNew10}, and we invoked \eqref{eqn33cc} for steps \eqref{eqnNew15} and \eqref{eqnNew16}. [For \eqref{eqnNew16}, $\lm_{\max}(G_{jk})\leq r(P,Q)$ for each $j,k$, so replacing each $\lm_{\max}(G_{jk})$ with $ r(P,Q)$ makes the overall operator larger.] Step~\eqref{eqnNew18} involves Lemma~\ref{thm1}, part (i).
\end{proof}

\xb
\section{Proof of Corollary \ref{thm6}}\label{aps2.5}
\xa

\begin{proof}
(i) Consider a channel ket $\ket{\Om}$ on $\HC_{abc}$ with $P=\{P_j\}$ and $Q=\{Q_k\}$ two POVMs on $a$, and apply \eqref{eqn57} and \eqref{eqn58} to $\ket{\Om}$:
\begin{align}
\label{eqn111}
\chi(P,b)&\leq H(P)-\log [1/\sqrt{r(P,P)}] \nonumber\\
\chi(P,b)+\chi(Q,c)&\leq H(P)+H(Q)-\log [1/r(P,Q)].
\end{align}
Now decompose $\ket{\Om}=(I_a\ot V)\ket{\Phi}$ as in Sec.~\ref{sbct2.2} and Fig.~\ref{fgr1}, where system $a'$ (of the same dimension as $a$) is introduced and fed into isometry $V$, and the state $\ket{\Phi} =(1/\sqrt{d_a}) \sum_j \ket{j}_a\ot\ket{j}_{a'}$ is maximally entangled, expanded here in the computational bases on $a$ and $a'$. By map-state duality \cite{Gri05}, think of $\ket{\Phi}$ as an isometry $\hat V$ from $\HC_a$ to $\HC_{a'}$, with $\hat V\ad\hat V =I_{a}$ \emph{and} $\hat V\hat V\ad=I_{a'}$ since $d_a=d_{a'}$. This means that $\widetilde P=\{\widetilde P_j\}=\{\hat V P_j \hat V\ad\} $ and $\widetilde Q=\{\widetilde Q_k\}=\{\hat V Q_k \hat V\ad \}$ are POVMs on $a'$. If outcome $P_j$ of $P$ occurs on $a$, then element $\widetilde P_j$ will get fed into the isometry $V$, so $\chi(P,b)=\chi(\widetilde P,\EC)$ and likewise $\chi(Q,c)=\chi(\widetilde Q,\FC)$, where $\EC$ and $\FC$ are the (complementary) channels to $b$ and $c$, respectively, associated with isometry $V$. Also, since $\rho_a=I_a/d_a$ for a channel ket, the probability for $P_j$ in \eqref{eqn6} given by $p_j=\Tr(P_j\rho_a)=\Tr(P_j)/d_a$ reduces to the corresponding formula in \eqref{eqn39}, so $H(P)=H(\widetilde P)$ and likewise $H(Q)=H(\widetilde Q)$. Finally, show that $r(\widetilde P, \widetilde Q)=r(P,Q)$ as follows:
\begin{align}
\label{eqn112}
&\| (\widetilde P_j)^{1/2}(\widetilde Q_k)^{1/2}\|^2_\infty =\lm_{\text{max}}[\hat V(Q_k)^{1/2}P_j(Q_k)^{1/2} \hat V\ad]\nonumber\\
&= \lm_{\text{max}}[(Q_k)^{1/2}P_j(Q_k)^{1/2}]= \| (P_j)^{1/2}(Q_k)^{1/2}\|^2_\infty,
\end{align}
where $\lm_{\text{max}}[\cdot]$ denotes the maximum eigenvalue and we used the fact that $(\widetilde Q_k)^{1/2}=\hat V(Q_k)^{1/2}\hat V\ad$, which follows from $[\hat V(Q_k)^{1/2}\hat V\ad]^2= \hat VQ_k \hat V\ad $ since $(Q_k)^{1/2}$ and $\hat V(Q_k)^{1/2} \hat V\ad$ are positive operators. Thus from \eqref{eqn111},
\begin{align}
\label{eqn113}
\chi(\widetilde P,\EC)&\leq H(\widetilde P)-\log [1/\sqrt{r(\widetilde P, \widetilde P)}]\nonumber\\
\chi(\widetilde P,\EC)+\chi(\widetilde Q,\FC)&\leq H(\widetilde P)+H(\widetilde Q)-\log [1/r(\widetilde P, \widetilde Q)].
\end{align}
Since $\hat V$ is a one-to-one mapping of the set of POVMs on $a$ to the set of POVMs on $a'$, then \eqref{eqn113} must be true for \emph{all} POVMs on $a'$, and one can replace $\widetilde P$ and $\widetilde Q$ with $P$ and $Q$ in \eqref{eqn113} for simplicity.

(ii) Equation~\eqref{eqn67} follows from \eqref{eqn66} since $H(v)=H(w)=\log d_a$ from \eqref{eqn39}.
\end{proof}

\xb
\section{Proof of Corollary \ref{thm7}}\label{aps2.6}
\xa

\begin{proof}
(i) For \eqref{eqn69}, let $b$ be a system that purifies $\rho_a$, apply \eqref{eqn58}, and by Theorem~\ref{thm2}, $\chi(N,b)=S(\rho_a)$. For \eqref{eqn70}, again let $b$ purify $\rho_a$, and apply \eqref{eqn57}. System $c$ is completely uncorrelated to $a$, so $H(P|c)=H(P)$, and by Theorem~\ref{thm2}, $\chi(N,b)=S(\rho_a)$. 

(ii) Equation \eqref{eqn71} follows from \eqref{eqn70} applied to MUBs $x$ and $y$:
\begin{equation}
\label{eqn114}
H(x)+H(y) \geq \log 2 +S(\rho_a).
\end{equation}
Denote $X$, $Y$, and $Z$ as the Pauli operators whose eigenvectors are the $x$, $y$, and $z$ bases. Consider a (possibly mixed) state in the $xy$ plane of the Bloch sphere: 
\begin{equation}
\label{eqn115}
\rho_a=(I_a+\alpha X+\beta Y)/2,
\end{equation}
for which $H(z)=\log 2$, so \eqref{eqn71} clearly holds for states of this form using \eqref{eqn114}. Now consider varying $\rho_a$ along a vertical path within the Bloch sphere, from the state $\rho_a$ (in the $xy$ plane) to a state $\rho'_a$ with some $z$ component but with the same $x$ and $y$ components:
\begin{equation}
\label{eqn116}
\rho'_a=(I_a+\alpha X+\beta Y+\gamma Z)/2,
\end{equation}
Denoting the relevant state with a subscript, note that $H(x)_{\rho_a}=H(x)_{\rho'_a}$ and $H(y)_{\rho_a}=H(y)_{\rho'_a}$ remain constant, so to prove \eqref{eqn71} for general states $\rho'_a$, we just need to show that $H(z)$ decreases more slowly than $S(\rho_a)$ along this path. This would be true if:
\begin{equation}
\label{eqn117}
H(z)_{\rho'_a}-S(\rho'_a) \geq H(z)_{\rho_a}-S(\rho_a)=\log 2-S(\rho_a).
\end{equation}
Due to the isotropic nature of the Bloch sphere, it is sufficient to check that \eqref{eqn117} holds for an initial state along the $x$-axis: $\rho_a=(I_a+\alpha X)/2$ and $\rho'_a=(I_a+\alpha X+\gamma Z)/2$, since $H(z)$ and $S(\rho_a)$ will vary in the same way along a vertical path regardless of an initial unitary rotation about $z$. But for such a state, $S(\rho_a)=H(x)_{\rho_a}=H(x)_{\rho'_a}$, and \eqref{eqn117} reduces to $H(z)_{\rho'_a}+ H(x)_{\rho'_a} \geq \log 2+S(\rho'_a)$, which is \eqref{eqn114} applied to MUBs $z$ and $x$. Thus, varying along a vertical path from a state in the $xy$ plane to a state with some $z$-component keeps the values of $H(x)$ and $H(y)$ constant, while not decreasing the value of $H(z)-S(\rho_a)$, proving the result in general.
\end{proof}

\xb
\section{Proof of Theorem \ref{thm8}}\label{aps2.7}
\xa

\begin{proof}

Let $c$ be a system that purifies $\rho_{ab}$. Re-write \eqref{eqn57} as
\begin{align}
\label{eqn121}
\chi(M,c)&\leq H(P|b)+H(M)+\log r(P,M),\nonumber\\
\chi(N,c)&\leq H(P|b)+H(N)+\log r(P,N).
\end{align}
Taken together, these two inequalities give an upper bound on the difference $|\chi(M,c)-\chi(N,c)|$. The difference is at most the one computed by allowing the $\chi$ quantity with the highest upper bound in \eqref{eqn121} to reach its bound, and allowing the other $\chi$ quantity to be zero. Thus,
\begin{align}
\label{eqn122}
&|\chi(M,c)-\chi(N,c)|\leq H(P|b)+ \nonumber\\
&\max \{H(M)+\log r(P,M),H(N)+\log r(P,N)\}.
\end{align}
By \eqref{eqn52}, substitute $b$ for $c$ on the left-hand-side. 

Rearranging \eqref{eqn121} to lower bound $H(M|c)$ and $H(N|c)$, and upper-bounding each respectively by $H(M)$ and $H(N)$, we can upper-bound their difference by the (maximum) difference between the upper bound of one and the lower bound of the other:
\begin{align}
\label{eqn123}
&|H(M|c)-H(N|c)|\leq H(P|b)+ \nonumber\\
&\max \{H(M)+\log r(P,N),H(N)+\log r(P,M)\}.
\end{align}
Again invoke \eqref{eqn52} to switch from $c$ to $b$ and obtain \eqref{eqn72}.

Now assuming $u$ and $v$ are MU with respect to $w$, \eqref{eqn73} follows from \eqref{eqn72} by setting $r(u,w)= r(v,w)=1/d_a$, and by noting that $H(u)\leq \log(d_a)$ and likewise for $H(v)$, so that the $\max \{\}$ term in \eqref{eqn72} is non-positive. 

Further specializing to the case of $H(w|b)=0$ and $v$ MU to $w$, then \eqref{eqn61} implies $H(v|c)=\log d_a$ and $\chi(v,c)=0$, and in turn that $\chi_K(v,c)=0$, because all $\chi_K$ measures are zero under the same conditions. Then by Theorem~\ref{thm3}, $H(v|b)=H(v|c)-\Dl\chi(b,c)= \log d_a+S(a|b)$, and $\chi_K(v,b)=\Dl \chi_K(b,c)=S_K(\rho_b)-S_K(\rho_{ab})$. 
\end{proof}

\xb
\section{Proof of Theorem \ref{thm10}}\label{aps2.8}
\xa

\begin{proof}
(i) First let $c$ purify $\rho_{ab}$, and by Theorem~\ref{thm3}, add the basis-invariant quantity $\Dl \chi(c,b)=H(N|b)-H(N|c)=S(\rho_c)-S(\rho_b)=S(a|b)$ to both sides of \eqref{eqn57}, setting $Q=N$, to obtain \eqref{eqn75}. Now to prove \eqref{eqn76}, let $cd$ purify $\rho_{ab}$ so that $\rho_{abc}=\Tr_d(\rho_{abcd})$ is a general (possibly mixed) state. Again by Theorem~\ref{thm3}, $[-S(a|b)]=\chi(M,b)-\chi(M,cd)\leq \chi(M,b)$ and $[-S(a|b)]=H(M|cd)-H(M|b)\leq H(M|cd) \leq H(M|c)$ by \eqref{eqn37}.

(ii) The argument for complementary quantum channels is the same. Add the basis-invariant quantity $\Dl \chi(\EC,\FC)$ to \eqref{eqn67} to obtain \eqref{eqn77}, and obtain \eqref{eqn78} using $\chi(u,\FC)\leqslant \log d_a-\Dl \chi(\EC,\FC)$.

(iii) Equation \eqref{eqn79} follows from $S(a\colo b)/2\geq -S(a|b)\geq \log d_a-[H(v|b)+H(w|b)]$. For \eqref{eqn80}, let $cd$ purify $\rho_{ab}$, then $H(v|b)+H(w|b)\geq \log d_a+S(a|b)\geq S(\rho_a)+S(a|b)=S(a\colo cd)\geq S(a\colo c)$.
\end{proof}

\xb
\section{Proof of Theorem \ref{thm11}}\label{aps2.9}
\xa

\begin{proof}
(i) First let us prove \eqref{eqn81} for pure $\rho_{abc}=\dya{\Om}$.
\begin{align}
\label{eqn124}
S(a\colo c)& = S(\rho_a)-\Dl \chi(b,c) \nonumber\\
&\leq H(N)-\log [1/\sqrt{r(N,N)}]-\Dl \chi(b,c) \nonumber\\
&= H(N|b)+\chi(N,c)-\log [1/\sqrt{r(N,N)}],
\end{align}
where the first line follows from Theorem~\ref{thm3}, and the second line is from \eqref{eqn69}. Now consider any $\rho_{abc}$. Apply the just-proven result \eqref{eqn124} to $\rho_{abcd}$:
\begin{equation}
\label{eqn125}
S(a\colo c)\leq H(N|bd)+\chi(N,c)-\log [1/\sqrt{r(N,N)}].
\end{equation}
where $\rho_{abcd}$ is a purification of $\rho_{abc}$. Then, \eqref{eqn81} is obtained by noting that $H(N|bd)\leq H(N|b)$ from \eqref{eqn37}.

If information about a rank-1 POVM $N$ is perfectly present in $b$, this implies that the elements of $N$ must be orthogonal and hence normalized, i.e.\ $N$ is some orthonormal basis $w=\{\ket{w_j}\}$. (The outputs $\rho_{bj}$ cannot all be orthogonal if the inputs are not orthogonal.) By the Truncation theorem of \cite{Gri07}, $\rho_{ac}=\sum_j p_j \dya{w_j}\ot \rho_{cj}$, i.e.\ $c$ is at-most classically correlated to the $w$ basis on $a$. Then the conditional density operators on $c$ ($\sigma_{ck}$ occurring with probability $q_k$) associated with POVM $P=\{P_k\}$ are related to those associated with the $w$ basis by $q_k \sigma_{ck} = \Tr_a(P_{k}\rho_{ac})=\sum_j M_{kj} p_j \rho_{cj}$, where $M_{kj}=\bra{w_j}P_{k}\ket{w_j}$. Now use the concavity of the entropy $S_K$ (all of our entropy functions have this property, see Sec.~\ref{sbct3.1}) and $\sum_k M_{kj}=1$ to show that:
\begin{align}
\label{eqn126}
\chi_K(P,c)&=S_K(\rho_c)-\sum_k q_k S_K(\sigma_{ck})\nonumber\\
&\leq S_K(\rho_c)-\sum_{k,j} M_{kj} p_j S_K(\rho_{cj})\nonumber\\
&= S_K(\rho_c)-\sum_j p_j S_K(\rho_{cj})=\chi_K(w,c).
\end{align}
The remark that $\rho_{ac}=\rho_a \ot \rho_c$ when all types are absent from $c$ seems obvious, although it is rigorously proven in Theorem 1 of \cite{Gri05}.

(ii) To prove \eqref{eqn83} for pure states, note that the right-hand-side of \eqref{eqn81} is an upper bound on $\chi(L,c)$ and $\chi(M,c)$ by \eqref{eqn42}, so it must also upper-bound their difference:
\begin{equation}
\label{eqn127}
|\chi(L,c)-\chi(M,c)| \leq \chi(N,c)+H(N|b)-\log[1/\sqrt{r(N,N)}].
\end{equation}
By \eqref{eqn52}, $b$ can replace $c$ on the left-hand-side.

In the case where information about $N$ is perfectly present in $b$ and absent from $c$, $\rho_{ac}=\rho_a\ot \rho_c$ by part (i) of this theorem, and $S(\rho_b)=S(\rho_{ac})=S(\rho_a)+S(\rho_c)$ by the additivity of $S$ for product states. Thus by Theorem~\ref{thm3}, for any rank-1 POVM $L$, $\chi(L,b)=\chi(L,b)-\chi(L,c)=S(\rho_b) -S(\rho_c)=S(\rho_a)$.

(iii) Equation \eqref{eqn85} follows immediately from $\chi(v,\FC) \leq \log d_a -\Dl \chi(\EC,\FC)$ and $\chi(v,\EC) \geq \Dl \chi(\EC,\FC)$, where  $\Dl \chi(\EC,\FC)= \chi(w,\EC)- \chi(w,\FC)$ is basis-invariant by Theorem~\ref{thm3}. In the extreme case where the $w$ type of information is perfectly present in $\EC$ and absent from $\FC$, $\Dl \chi(\EC,\FC)=\log d_a$, hence $\chi(v,\FC)=0$ and $\chi(v,\EC)=\log d_a$ for \emph{all} $v$.
\end{proof}

\xb
\xa

\bibliographystyle{unsrt}
\bibliography{infosplitting}

\begin{thebibliography}{10}

\bibitem{CvTh06}
T.~M. Cover and J.~A. Thomas.
\newblock {\em Elements of Information Theory}.
\newblock Wiley, New York, 2nd edition, 2005.

\bibitem{Gri07}
Robert~B. Griffiths.
\newblock Types of quantum information.
\newblock {\em Phys. Rev. A}, 76:062320, 2007.

\bibitem{CQT}
Robert~B. Griffiths.
\newblock {\em Consistent Quantum Theory}.
\newblock Cambridge University Press, Cambridge, 2002.

\bibitem{GriNiu96}
Robert~B. Griffiths and Chi-Sheng Niu.
\newblock Semiclassical fourier transform for quantum computation.
\newblock {\em Phys. Rev. Lett.}, 76:3228--3231, 1996.

\bibitem{HHHH09}
Ryszard Horodecki, Pawe\l\ Horodecki, Micha\l\ Horodecki, and Karol Horodecki.
\newblock Quantum entanglement.
\newblock {\em Rev. Mod. Phys.}, 81:865--942, 2009.

\bibitem{NieChu00}
Michael~A. Nielsen and Isaac~L. Chuang.
\newblock {\em Quantum Computation and Quantum Information}.
\newblock Cambridge University Press, Cambridge, 5th edition, 2000.

\bibitem{BertaEtAl}
M.~{Berta}, M.~{Christandl}, R.~{Colbeck}, J.~M. {Renes}, and R.~{Renner}.
\newblock {The uncertainty principle in the presence of quantum memory}.
\newblock {\em Nature Physics}, 6:659, 2010.

\bibitem{RenesBoileau}
Joseph~M. Renes and Jean-Christian Boileau.
\newblock Conjectured strong complementary information tradeoff.
\newblock {\em Phys. Rev. Lett.}, 103:020402, 2009.

\bibitem{TomRen2010}
M.~{Tomamichel} and R.~{Renner}.
\newblock {Uncertainty Relation for Smooth Entropies}.
\newblock {\em Phys. Rev. Lett.}, 106:110506, 2011.

\bibitem{DevWin03}
I.~Devetak and A.~Winter.
\newblock Classical data compression with quantum side information.
\newblock {\em Phys. Rev. A}, 68:042301, 2003.

\bibitem{Prs90b}
Asher Peres.
\newblock Neumark's theorem and quantum inseparability.
\newblock {\em Found. Phys.}, 20:1441--1453, 1990.

\bibitem{JozsaEtAl03}
R.~{Jozsa}, M.~{Koashi}, N.~{Linden}, S.~{Popescu}, S.~{Presnell},
  D.~{Shepherd}, and A.~{Winter}.
\newblock {Entanglement cost of generalised measurements}.
\newblock e-print arXiv:quant-ph/0303167.

\bibitem{Prs93}
Asher Peres.
\newblock {\em Quantum Theory: Concepts and Methods}.
\newblock Kluwer Academic Publishers, Dordrecht, The Netherlands, 1993.

\bibitem{Gri05}
Robert~B. Griffiths.
\newblock Channel kets, entangled states, and the location of quantum
  information.
\newblock {\em Phys. Rev. A}, 71:042337, 2005.

\bibitem{BenZyc06}
Ingemar Bengtsson and Karol \.Zyczkowski.
\newblock {\em Geometry of Quantum States}.
\newblock Cambridge University Press, Cambridge, 2006.

\bibitem{ChristWinterIEEE2005}
M.~Christandl and A.~Winter.
\newblock Uncertainty, monogamy, and locking of quantum correlations.
\newblock {\em Information Theory, IEEE Transactions on}, 51:3159--3165, 2005.

\bibitem{SubaddivityQ}
Koenraad M.~R. Audenaert.
\newblock Subadditivity of $q$-entropies for $q>1$.
\newblock {\em Journal of Mathematical Physics}, 48:083507, 2007.

\bibitem{VedralReview02}
V.~Vedral.
\newblock The role of relative entropy in quantum information theory.
\newblock {\em Rev. Mod. Phys.}, 74:197--234, 2002.

\bibitem{OhPe93}
M.~Ohya and D.~Petz.
\newblock {\em Quantum Entropy and Its Use}.
\newblock Springer, 1st edition, 1993.

\bibitem{FuchsThesis}
C.~A. {Fuchs}.
\newblock {Distinguishability and Accessible Information in Quantum Theory},
  1996.
\newblock e-print arXiv:quant-ph/9601020.

\bibitem{SWW96}
Benjamin Schumacher, Michael Westmoreland, and William~K. Wootters.
\newblock Limitation on the amount of accessible information in a quantum
  channel.
\newblock {\em Phys. Rev. Lett.}, 76:3452--3455, 1996.

\bibitem{WuEtAl2009}
Shengjun Wu, Uffe~V. Poulsen, and Klaus M\o{}lmer.
\newblock Correlations in local measurements on a quantum state, and
  complementarity as an explanation of nonclassicality.
\newblock {\em Phys. Rev. A}, 80:032319, Sep 2009.

\bibitem{Bhatia}
Rajendra Bhatia.
\newblock {\em Matrix Analysis}.
\newblock Springer-Verlag, New York, 1997.

\bibitem{ColesEtAlTBP}
P.~J. {Coles}, L.~{Yu}, and M.~{Zwolak}.
\newblock {Relative entropy derivation of the uncertainty principle with
  quantum side information}.
\newblock e-print arXiv:1105.4865 [quant-ph].

\bibitem{KrishnaParth}
M.~{Krishna} and K.~{Parthasarathy}.
\newblock {An Entropic Uncertainty Principle for Quantum Measurements}.
\newblock {\em Indian J. of Statistics}, Ser. A 64:842, 2002.

\bibitem{SanchezRuiz1995}
Jorge S{\'a}nchez-Ruiz.
\newblock Improved bounds in the entropic uncertainty and certainty relations
  for complementary observables.
\newblock {\em Physics Letters A}, 201:125 -- 131, 1995.

\bibitem{Schu98}
Benjamin Schumacher and Michael~D. Westmoreland.
\newblock Quantum privacy and quantum coherence.
\newblock {\em Phys. Rev. Lett.}, 80:5695--5697, 1998.

\bibitem{EURreview1}
S.~{Wehner} and A.~{Winter}.
\newblock {Entropic uncertainty relations -- a survey}.
\newblock {\em New Journal of Physics}, 12:025009, 2010.

\bibitem{MaassenUffink}
Hans Maassen and J.~B.~M. Uffink.
\newblock Generalized entropic uncertainty relations.
\newblock {\em Phys. Rev. Lett.}, 60:1103--1106, 1988.

\bibitem{Hall1}
Michael J.~W. Hall.
\newblock Information exclusion principle for complementary observables.
\newblock {\em Phys. Rev. Lett.}, 74:3307--3311, 1995.

\bibitem{CerfEtAl}
Nicolas~J. Cerf, Mohamed Bourennane, Anders Karlsson, and Nicolas Gisin.
\newblock Security of quantum key distribution using $d$-level systems.
\newblock {\em Phys. Rev. Lett.}, 88:127902, 2002.

\bibitem{Spekkens08}
Dennis Kretschmann, David~W. Kribs, and Robert~W. Spekkens.
\newblock Complementarity of private and correctable subsystems in quantum
  cryptography and error correction.
\newblock {\em Phys. Rev. A}, 78:032330, 2008.

\bibitem{BenyOreshkov2010}
C\'edric B\'eny and Ognyan Oreshkov.
\newblock General conditions for approximate quantum error correction and
  near-optimal recovery channels.
\newblock {\em Phys. Rev. Lett.}, 104:120501, 2010.

\bibitem{HW2010}
P.~{Hayden} and A.~{Winter}.
\newblock {The Fidelity Alternative and Quantum Identification}.
\newblock e-print arXiv:1003.4994v3 [quant-ph].

\bibitem{PrevedelEtAl2010}
{R. {Prevedel}, D.~R. {Hamel}, R. {Colbeck}, K. {Fisher}, and K.~J. {Resch}}.
\newblock {Experimental investigation of the uncertainty principle in the
  presence of quantum memory}.
\newblock e-print arXiv:1012.0332 [quant-ph].

\end{thebibliography}

\xb

\end{document}